\documentclass[11pt]{article}

\usepackage[english]{babel}
\usepackage[utf8x]{inputenc}
\usepackage[T1]{fontenc}

\usepackage[skip=2pt]{caption}

\usepackage{amsmath}
\usepackage{amssymb}
\usepackage{graphicx}
\usepackage{subcaption}
\usepackage[colorinlistoftodos]{todonotes}
\usepackage{xspace}
\usepackage{bm}
\usepackage{algorithm}
\usepackage[noend]{algorithmic}
\usepackage[inline]{enumitem}
\usepackage{balance}
\usepackage{fullpage}
\usepackage{url}

\usepackage{davids-macros/color-edits}
\usepackage{davids-macros/definitions}
\addauthor[Kartik]{kl}{blue}
\addauthor[David]{dk}{green}

\newcommand{\e}{\mathrm{e}}

\newcommand{\OPT}[1][]{\ensuremath{%
\ifthenelse{\equal{#1}{}}{\text{OPT}}{\text{OPT}^{(#1)}}}\xspace}
\newcommand{\OPTH}[1][]{\ensuremath{%
\ifthenelse{\equal{#1}{}}{\hat{\text{OPT}}}{\hat{\text{OPT}}^{(#1)}}}\xspace}

\newcommand{\INFFUN}[2][]{\ensuremath{%
\ifthenelse{\equal{#1}{}}{f_{#2}}{f^{(#1)}_{#2}}}\xspace}
\newcommand{\INFFUNP}[2][]{\ensuremath{%
\ifthenelse{\equal{#1}{}}{\hat{f}_{#2}}{\hat{f}^{(#1)}_{#2}}}\xspace}
\newcommand{\InfFun}[3][]{\ensuremath{\INFFUN[#1]{#2}(#3)}\xspace}
\newcommand{\InfFunP}[3][]{\ensuremath{\INFFUNP[#1]{#2}(#3)}\xspace}

\newcommand{\Thold}[2][]{\ensuremath{%
\ifthenelse{\equal{#1}{}}{\theta_{#2}}{\theta^{(#1)}_{#2}}}\xspace}

\newcommand{\Active}[2][]{\ensuremath{%
\ifthenelse{\equal{#1}{}}{A_{#2}}{A^{(#1)}_{#2}}}\xspace}
\newcommand{\ActiveUN}[2][]{\ensuremath{%
\ifthenelse{\equal{#1}{}}{AU_{#2}}{AU^{(#1)}_{#2}}}\xspace}

\newcommand{\FinalActive}[1][]{\ensuremath{%
\ifthenelse{\equal{#1}{}}{\hat{A}}{\hat{A}^{(#1)}}}\xspace}

\newcommand{\FinalSize}[2][]{\ensuremath{%
\ifthenelse{\equal{#1}{}}{\sigma(#2)}{\sigma^{(#1)}(#2)}}\xspace}
\newcommand{\FINALSIZE}[1][]{\ensuremath{%
\ifthenelse{\equal{#1}{}}{\sigma}{\sigma^{(#1)}}}\xspace}
\newcommand{\FinalSizeMarg}[2][]{\ensuremath{%
\ifthenelse{\equal{#1}{}}{\sigma_m(#2)}{\sigma_m^{(#1)}(#2)}}\xspace}

\newcommand{\Camp}[1][]{\ensuremath{%
\ifthenelse{\equal{#1}{}}{C}{C^{(#1)}}}\xspace}

\newcommand{\VALFUN}[1][]{\ensuremath{%
\ifthenelse{\equal{#1}{}}{w}{w^{(#1)}}}\xspace}
\newcommand{\ValFun}[2][]{\ensuremath{%
\ifthenelse{\equal{#1}{}}{w(#2)}{w^{(#1)}(#2)}}\xspace}
\newcommand{\PVALFUN}[1][]{\ensuremath{%
\ifthenelse{\equal{#1}{}}{\hat{w}}{\hat{w}^{(#1)}}}\xspace}

\newcommand{\PValFun}[2][]{\ensuremath{%
\ifthenelse{\equal{#1}{}}{\hat{w}(#2)}{\hat{w}^{(#1)}(#2)}}\xspace}
\newcommand{\ValFunMarg}[2][]{\ensuremath{%
\ifthenelse{\equal{#1}{}}{w_m(#2)}{w_m^{(#1)}(#2)}}\xspace}
\newcommand{\PValFunMarg}[2][]{\ensuremath{%
\ifthenelse{\equal{#1}{}}{\hat{w_m}(#2)}{\hat{w_m}^{(#1)}(#2)}}\xspace}

\newcommand{\BUDGET}[1][]{\ensuremath{%
\ifthenelse{\equal{#1}{}}{B}{B^{(#1)}}}\xspace}
\newcommand{\RBUDGET}[1][]{\ensuremath{%
\ifthenelse{\equal{#1}{}}{B}{B_R^{(#1)}}}\xspace}

\newcommand{\SBUDGET}[1][]{\ensuremath{%
\ifthenelse{\equal{#1}{}}{K}{K^{(#1)}}}\xspace}
\newcommand{\PER}[1][]{\ensuremath{%
\ifthenelse{\equal{#1}{}}{\gamma}{\gamma^{(#1)}}}\xspace}

\newcommand{\EMOBJ}{\ensuremath{W}\xspace}
\newcommand{\EMObj}[1]{\ensuremath{\EMOBJ(#1)}\xspace}
\newcommand{\MEOBJ}{\ensuremath{\bar{W}}\xspace}
\newcommand{\MEObj}[1]{\ensuremath{\MEOBJ(#1)}\xspace}

\newcommand{\ExposureBound}[1]{\ensuremath{r_{#1}}\xspace}

\newcommand{\EdgeProb}[2][]{\ensuremath{%
\ifthenelse{\equal{#1}{}}{p_{#2}}{p^{(#1)}_{#2}}}\xspace}
\newcommand{\EdgeProbP}[2][]{\ensuremath{%
\ifthenelse{\equal{#1}{}}{\hat{p}_{#2}}{\hat{p}^{(#1)}_{#2}}}\xspace}

\newcommand{\GRAPH}[1][]{\ensuremath{%
\ifthenelse{\equal{#1}{}}{G}{G^{(#1)}}}\xspace}
\newcommand{\Inter}[2][]{\ensuremath{%
\ifthenelse{\equal{#1}{}}{\alpha(#2)}{A^{(#1)}(#2)}}\xspace}

\newcommand{\EXPGRAPH}{\ensuremath{\hat{\mathcal{G}}}\xspace}
\newcommand{\NODES}[1][]{\ensuremath{%
\ifthenelse{\equal{#1}{}}{U}{U^{(#1)}}}\xspace}
\newcommand{\RRP}[1][]{\ensuremath{%
\ifthenelse{\equal{#1}{}}{\mathcal{R}}{\mathcal{R}^{(#1)}}}\xspace}

\newcommand{\Select}[2][]{\ensuremath{%
\ifthenelse{\equal{#1}{}}{z_{#2}}{z^{(#1)}_{#2}}}\xspace}
\newcommand{\ISelect}[2][]{\ensuremath{%
\ifthenelse{\equal{#1}{}}{\hat{z}_{#2}}{\hat{z}^{(#1)}_{#2}}}\xspace}
\newcommand{\SelectV}{\ensuremath{\bm{z}}\xspace}

\newcommand{\SSel}[1]{\ensuremath{x_{#1}}\xspace}
\newcommand{\ISSel}[1]{\ensuremath{\hat{x}_{#1}}\xspace}
\newcommand{\SSelV}{\ensuremath{\bm{x}}\xspace}

\newcommand{\Rev}[1][]{\ensuremath{%
\ifthenelse{\equal{#1}{}}{y}{y^{(#1)}}}\xspace}
\newcommand{\IRev}[1][]{\ensuremath{%
\ifthenelse{\equal{#1}{}}{\hat{y}}{\hat{y}^{(#1)}}}\xspace}
\newcommand{\RevV}{\ensuremath{\bm{y}}\xspace}

\newcommand{\Facebook}{\textit{Facebook}\xspace}
\newcommand{\DBLP}{\textit{DBLP}\xspace}
\newcommand{\NetHEPT}{\textit{NetHEPT}\xspace}
\newcommand{\Advogato}{\textit{Advogato}\xspace}
\newcommand{\Pokec}{\textit{Pokec}\xspace}
\newcommand{\Flixster}{\textit{Flixster}\xspace}
\newcommand{\Greedy}{\textit{Greedy}\xspace}
\newcommand{\LPRound}{\textit{LP-Rounding}\xspace}
\newcommand{\MaxDeg}{\textit{Max-Degree}\space}
\newcommand{\EigenCent}{\textit{Eigen-Centrality}\xspace}
\newcommand{\LPOPT}{\ensuremath{\mathrm{OPT}_{\mathrm{LP}}}\xspace}

\begin{document}
\title{Approximation Algorithms for Coordinating Ad Campaigns on
  Social Networks}
\author{Kartik Lakhotia \qquad David Kempe\\
University of Southern California\\
klakhoti@usc.edu, david.m.kempe@gmail.com}

\maketitle

\begin{abstract}
	We study a natural model of coordinated social ad campaigns over a
social network, based on models of Datta et al.~and Aslay et al.
Multiple advertisers are willing to pay the host --- up to a known budget ---
per user exposure, whether that exposure is sponsored or organic
(i.e., shared by a friend). 
Campaigns are seeded with sponsored ads to some users,
but no network user must be exposed to too many sponsored ads.
As a result, while ad campaigns proceed independently over the network,
they need to be carefully coordinated with respect to their seed sets.

We study the objective of maximizing the network's
total ad revenue.
Our main result is to show that 
under a broad class of social influence models, 
the problem can be reduced to maximizing a submodular function subject
to two matroid constraints; it can therefore be approximated within 
a factor essentially \half in polynomial time.
When there is no bound on the individual seed set sizes of
advertisers, the constraints correspond only to a single matroid,
and the guarantee can be improved to $1-1/\e$;
in that case, a factor \half is achieved by a practical greedy algorithm.
The $1-1/\e$ approximation algorithm for the matroid-constrained
problem is far from practical; however, we show that specifically
under the Independent Cascade model,
LP rounding and Reverse Reachability techniques can be combined
to obtain a $1-1/\e$ approximation algorithm which scales to several
tens of thousands of nodes.

Our theoretical results are complemented by experiments evaluating the
extent to which the coordination of multiple ad campaigns inhibits the
revenue obtained from each individual campaign, as a function of the
similarity of the influence networks and the strength of ties in the network.
Our experiments suggest that as networks for different advertisers
become less similar, the harmful effect of competition decreases.
With respect to tie strengths, we show that the most harm is done in
an intermediate range.

\end{abstract}

\section{Introduction} \label{sec:introduction}
Advertising is the most important and successful business model among
social network sites.
It is widely believed that the adoption of products has a significant
social component to it, wherein people recommend products to each
other, seek each other's opinion on products, or simply see others
use a product.
Under advertising campaigns on social networks,
users are shown some ads embedded in their newsfeed from friends;
the advertisers' hope is that the users will voluntarily share the
ads, or adopt the product and be observed by their friends using it.

While models and algorithms for viral marketing have a long history of
study (e.g., \cite{domingos:richardson,%
goldenberg:libai:muller:complex,%
goldenberg:libai:muller:talk,%
hartline:mirrokni:sundarajan,%
InfluenceSpread}, and
\cite{chen:lakshmanan:castillo:influence-maximization-book} for a survey),
the focus in all of these works (and a much larger body of literature)
has been on campaigns to maximize the reach of a single product.
When multiple products have been considered, it has typically been in
the context of competition, wherein each member of the network adopts
at most one product (e.g., \cite{dubey:garg:demeyer:competing,%
goyal:kearns:competitive,%
GoyalKearnsCompetitive});
though see \cite{lu:chen:lakshmanan:complementarity} for a model
wherein products may exhibit complementarities.

Another very important way in which multiple products' advertising
campaigns interact was articulated in two works
by Datta, Majumder, and Shrivastava~\cite{datta:majumder:shrivastava}
and by Aslay, Lu, Bonchi, Goyal, and Lakshmanan \cite{aslay2015viral}.
Both groups of authors observe that marketing campaigns interact in
that they involve displaying ads to the same set of users in a social network,
and users should not be shown too many ads,
lest they feel too inundated with ads and leave the social
network~\cite{lin2014steering}.%
\footnote{According to the 2012 Digital Advertising Attitudes
Report~\cite{DAA2012report},  
users may even exhibit negative responses towards products if they
receive too many advertisements.}

More precisely, both sets of authors assume (explicitly or implicitly)
that product exposure comes in two forms: sponsored or user-shared.
When a user's friend willingly shares an ad, or the user
observes a friend using a product, this is perceived as genuine content.
On the other hand, a sponsored ad shown by the network itself is
viewed as advertising.
Both Datta et al.~\cite{datta:majumder:shrivastava}
and Aslay et al.~\cite{aslay2015viral} thus impose a
constraint on the number of \emph{sponsored} ads that a user can be shown:
the number of sponsored ads shown to node $v$ must not exceed a given
bound \ExposureBound{v}.
While only sponsored ads are deemed ``annoying,''
advertisers do profit equally from sponsored or user-shared exposure.
Imposing tight constraints on the number of sponsored ads
  across ad campaigns,
  and solving the corresponding optimization problem,
  should thus reduce the exposure of users to annoying sponsored ads.
  If the coordinated optimization problem is solved well, then this
  reduction might be achieved without much reduction in network profit.%
\footnote{Whether users should be shown ads at all on social networks
  is an ethical and normative question beyond the scope of our work.
  We take no position on whether social networking platforms
  should be considered public utilities, financed by subscription
  fees, or by ad revenue. Our goal here is to limit users' exposure to
  sponsored ads while still providing high coverage to individual
  advertisers in the currently practiced model.}
  

Both Datta et al.~and Aslay et al.~study the network's (here also
called \emph{host}) goal of maximizing its own payoff (also called
\emph{revenue}), by carefully coordinating the ad campaigns of
multiple advertisers.
In both models, the revenue that can be obtained from an advertiser
$j$ is constrained by the advertiser's ``budget.''
In the model of Datta et al., the advertiser pays the host for each
user exposure (sponsored or user-shared);
however, the number of sponsored ads for advertiser $j$ is bounded by
some number \SBUDGET[j].
In the standard language of influence maximization (precise
definitions are given in Section~\ref{sec:problem-statement}),
this translates to constraints on the \emph{seed set sizes} of each
advertiser.
In contrast, Aslay et al.~assume that the advertiser's budget
constrains the total payment for exposures;
if an advertiser with budget \BUDGET[j] receives $N$ exposures and is
willing to pay \PER[j] per exposure, the network's revenue is
$\min(\BUDGET[j], \PER[j] \cdot N)$.

The advertisers' hard budget constraints in the model of Aslay et
al.~raise an interesting question: 
what happens if $\PER[j] \cdot N > \BUDGET[j]$, i.e.,
an advertising cascade reaches \emph{more} network users
than the advertiser is willing to pay for?
Aslay et al.~\cite{aslay2015viral} posit that the network views this
negatively, and incurs a \emph{regret} over giving the advertiser free
exposure.\footnote{A possible justification is that if this happens
  repeatedly, advertisers may learn to lower their declared budget and
  pay less.}
In their model, the network obtains regret for \emph{undershooting}
the budget as well, this time over lost potential revenue.
Their goal is then to minimize the total regret.
Because the regret is 0 only when the network exactly hits the
advertisers' budgets,
a straightforward reduction from \textsc{Partition} style problems
shows not only NP-hardness, but also hardness of any multiplicative
approximation.
Under the same model, Tang and Yuan~\cite{tang2016optimizing}
consider a modified payoff function instead of minimizing regret.
Their objective function consists of the net revenue generated (capped by the budget),
with an additional negative term for excess exposure; this term penalizes
overshooting.
Notice that over- or undershooting the budget by a small amount in
this model does not have as severe consequences for approximation
guarantees as it does for the regret objective.
As a result, Tang and Yuan achieve a $\frac{1}{4}$-approximation using a
greedy approach. 

In this article, like Tang and Yuan, we depart from the regret
objective of Aslay et al. 
In particular, we do not believe that a network incurs ``regret'' from
accidentally giving advertisers free exposure.%
\footnote{Such regret may well be a psychological reality for humans,
but it appears less applicable at the level of a major company.}
One justification for a negative term for overshooting is that
  advertisers may learn to underbid, waiting for the free higher exposure.
However, an equally strong case could be made that
giving advertisers free exposure is beneficial,
in that it makes the advertisers more likely to run future
ad campaigns on the site,
which in turn may lead to an increase, not decrease, in future revenue.
Once a regret for overshooting the advertiser's budget is not a
concern any more,
a much more natural objective is to maximize the network's total
revenue, which is the objective function we study here.
(A precise definition of the model is given in
Section~\ref{sec:problem-statement}.)

Our main result --- in Section~\ref{sec:general} --- is a general
treatment of optimization in this model, 
subsuming both versions of budgets.
We consider very general influence models
(potentially different for each advertiser),
and only require that the local influence function at every node $v$
for every individual advertiser $j$ be monotone submodular%
\footnote{Recall that a set function $f$ is \emph{submodular} iff
  $f(S \cup \SET{x}) - f(S) \geq f(T \cup \SET{x}) - f(T)$
  whenever $S \subseteq T$.}.
Under these assumptions, we show that 
the network's revenue can be approximated to within constant factors.
More specifically, when there are no constraints on individual
advertisers' seed set sizes (only budget constraints limiting the
total number of exposures the advertiser is willing to pay for),
then there is a polynomial-time $(1-1/\e)$-approximation algorithm,
and a simple greedy algorithm achieves a \half-approximation.
When in addition, there are constraints on the seed set sizes of
individual advertisers, we obtain a $\frac{1}{2+\epsilon}$
approximation, for every $\epsilon > 0$.

We show these results by expressing the optimization problem as a
monotone submodular maximization problem over suitably chosen
matroid domains, much in the style of Vondr\'{a}k et al.'s result
on welfare-maximizing partitions
\cite{calinescu:chekuri:pal:vondrak,vondrak:submodular-welfare}.
As observed by Tang and Yuan as well,
when there are no constraints on individual seed set sizes,
the constraints jointly define a matroid.
This allows us to bring to bear on the problem Vondr\'{a}k's
polynomial-time $(1-1/\e)$-approximation algorithm for maximizing a
submodular function subject to a matroid constraint
\cite{calinescu:chekuri:pal:vondrak,vondrak:submodular-welfare}
(see also subsequent work
\cite{chekuri:vondrak:zenklusen,gharan:vondrak:annealing,vondrak:approximability}),
and the much simpler \half-approximation algorithm due to Fischer et
al.~\cite{fischer:nemhauser:wolsey} (see also \cite{lehmann:lehmann:nisan}).
When there are additional constraints on individual seed set sizes, we show that
the feasible ad campaigns can be expressed as the intersection of
\emph{two} matroids, and the algorithm of
Lee et al.~\cite{lee:sviridenko:vondrak:multiple-matroids}
gives a $\frac{1}{2+\epsilon}$ approximation.

While we do not develop new algorithms or heuristics in this
part of our work, we consider it an important contribution to
explicitly reduce from advertiser competition constraints to the
intersections of matroids.
Such a characterization was missing from prior work, resulting in
heuristics with weaker guarantees and rederiving known arguments.
Our work builds on the work of Tang and Yuan and generalizes it to
provide a clean optimization framework that allows the
incorporation of different constraints in the future, and the
application of known powerful optimization techniques and
guarantees.

While Vondr\'{a}k's beautiful $(1-1/\e)$-approximation algorithm for
maximizing a submodular function under a matroid constraint
\cite{vondrak:submodular-welfare} runs in polynomial time,
it is not practical.%
\footnote{The running time is roughly $O(n^8)$, with large constants.
  To the best of our knowledge, the full algorithm has never been
  implemented, but we would be surprised if it scaled to more than 10 nodes.}
With this concern in mind, in Section~\ref{sec:IC},
we develop a $(1-1/\e)$-approximation algorithm which scales to
moderately sized networks comprising several tens of thousands of nodes.
Unlike our main results, this algorithm is not fully general,
providing guarantees only for the Independent Cascade (IC) Model,
and satisfying the budget constraint only in expectation,
rather than for each solution.
Networks of tens of thousands of nodes are of practical interest:
they are often encountered in area-aware advertising \cite{HPLOCAL2018, FBREACH}. 
Another natural domain with small networks is targeted health and 
social intervention programs \cite{WOHLWPWTR}, where network sizes are often only in 
the hundreds or thousands; it is very plausible that a principal may want to coordinate 
multiple interventions, such as safe drug usage, safe sex practices, and others.

The reason that the algorithm only works for the IC model
is that it is based on randomly sampling Reverse
Reachable (RR) sets~\cite{borgs:brautbar:chayes:lucier,%
tang:shi:xiao:martingale,tang:xiao:shi:near-optimal}.
These RR sets are then used to define and solve
a generalized Maximum Coverage LP,
which is then rounded using techniques of Gandhi et
al.~\cite{gandhi:khuller:parthasarathy:srinivasan:dependent}.
An additional important benefit of the algorithm is that the 
fractional LP which it rounds provides an upper bound on the
performance of the \emph{optimum} solution.
This upper bound is useful in our experimental evaluation,
not only of the LP rounding algorithm, but of other algorithms as well.

In Section~\ref{sec:experiments}, we describe experiments to evaluate
the effects of competition and tie strengths on coordinated
advertising campaigns.
A comparison of algorithms shows that while the greedy algorithm has a
worse worst-case approximation guarantee than LP rounding, it performs
(marginally) better in experiments; both algorithms beat several
natural heuristics, and get within 85\% of the LP-based upper bound on the
optimum; this is significantly better than their worst-case guarantees.
We also develop a parallel version of the Greedy algorithm that 
accelerates overall processing by a factor of 12 on a 36-core machine.
This allows us to scale to graphs with
millions of nodes and dozens of advertisers.

Since approximation algorithms for ``standard'' Influence Maximization
have been extensively studied in the literature,
our main interest in the experimental evaluation is the effect of
competition between products, and the interaction between competition
and tie strengths in the network.
Our experiments show that the more similar the advertisers'
influence networks, the more the host's payoff decreases compared to
the sum of what could be extracted from the advertisers in isolation.
This effect is exacerbated when the network has a small number of
highly influential nodes, since each advertiser will only be able to
target few of these nodes.
When nodes' influence is more even across the network,
the host loses less revenue, because there are enough different
``parts'' of the network to extract revenue from all advertisers.
When ties are very weak or very strong, the effect of competition
is again attenuated, while for intermediate tie strengths,
competition can lead to a significant revenue loss compared to
treating each advertiser separately.

\section{Problem Statement} \label{sec:problem-statement}
There are $m$ advertisers $j=1, \ldots, m$ aiming to advertise on a
network of $n$ nodes/individuals $V = \SET{1, \ldots, n}$.
We will use $u$ and $v$ (and their variants) to denote nodes,
and $j$ (and variants) exclusively for advertisers.

Product information (or ads or influence) propagates through the
social network according to the \emph{general threshold model}
\cite{InfluenceSpread,mossel:roch:submodular}, defined as follows:
For each node $v$ and advertiser $j$,
there is a monotone and submodular \emph{local influence function}
$\INFFUN[j]{v} : V \to [0,1]$ with $\INFFUN[j]{v}(\emptyset) = 0$.
Each node $v$ independently chooses thresholds
$\Thold[j]{v} \in [0,1]$ uniformly at random for each product $j$.
Let \Active[j]{t} be the set of nodes that have shared ad $j$ by round $t$.
(We call such nodes \emph{active for ad $j$}.)
Node $v$ becomes active for ad $j$ in round $t+1$ iff
$\InfFun[j]{v}{\Active[j]{t}} \geq \Thold[j]{v}$.
The (random) process is seeded with \emph{seed sets} \Active[j]{0}
(subject to constraints discussed below).
It quiesces when no new activations occur in a single round.
At that point, for each ad $j$, a final (random) set \FinalActive[j]
has become active.
The general threshold model subsumes most standard models of influence
spread, including the Independent Cascade and Linear Threshold models.
Notice that
(1) the diffusions for different $j$ proceed independently,
except for joint constraints on the seed sets (discussed below), and
(2) different ads can in principle follow different diffusion models.

Each advertiser $j$ has a non-negative and monotone
value function \VALFUN[j] giving the payment that $j$ will make
to the network site as a function of the active nodes in the end.
In full generality, this function could depend on the (random) \emph{set}
\FinalActive[j] of nodes that are active for $j$ in the end;
in this case, we will assume that \VALFUN[j] is also submodular.
In most cases, however, all nodes in the network will have the same
value to advertiser $j$, in which case \VALFUN[j] will only depend on
the cardinality $|\FinalActive[j]|$.
In this case, we will assume that \VALFUN[j] is (weakly) concave;
notice that this is a special case of the previous one, because viewed
as a function of \FinalActive[j] (rather than $\SetCard{\FinalActive[j]}$),
the function \VALFUN[j] is submodular.
For notational convenience, and in keeping with much of the prior
literature, we write
$\FinalSize[j]{\Active[j]{0}} = \SetCard{\FinalActive[j]}$
for the (random) number of nodes that are active in the end when
the process starts with the node set \Active[j]{0}.

A particularly natural case
--- closest to the definition of Aslay et al.~\cite{aslay2015viral} --- 
is when $\ValFun[j]{x} = \min(\BUDGET[j], \PER[j] \cdot x)$.
Here, \BUDGET[j] is the advertiser's budget,
while \PER[j] is the amount he is willing to pay per exposure.

The network's goal is to choose seed sets \Active[j]{0} for all
advertisers $j$,
subject to additional constraints discussed below,
so as to maximize one of the following two global revenue functions:

\begin{align}
\EMObj{\Active[1]{0}, \Active[2]{0}, \ldots, \Active[m]{0}}
  & = \sum_{j=1}^m \Expect{\ValFun[j]{\FinalSize[j]{\Active[j]{0}}}},
\label{eqn:expect-outside}
\\
\MEObj{\Active[1]{0}, \Active[2]{0}, \ldots, \Active[m]{0}}
& = \sum_{j=1}^m \ValFun[j]{\Expect{\FinalSize[j]{\Active[j]{0}}}}.
\label{eqn:expect-inside}
\end{align}

Notice the subtle difference between the two definitions:
\EMOBJ captures the expected revenue from advertisers who are charged
for each campaign individually, according to their functions \VALFUN[j].
The objective \EMOBJ extends straightforwardly to the case where \VALFUN[j]
depends not only on the cardinality of \FinalActive[j],
but on the specific set.
In contrast, \MEOBJ corresponds to the case in which advertisers
are charged according to the expected exposure.
An essentially equivalent way of expressing the same objective
(up to random noise) is that the revenue is based on the
\emph{average} exposure for advertiser $j$ over a large number of
campaigns.

Specifically in the context of the linear revenue function with a cap
of the budget \BUDGET[j],
this means that under \EMOBJ, advertiser $j$ is \emph{never} charged more
than \BUDGET[j] for any individual campaign,
while under \MEOBJ, advertiser $j$'s budget is not exceeded on
average over multiple campaigns.

\smallskip

The main constraint, common to both the models of Aslay et al.~and
Datta et al., is that each node can be exposed to only a limited
number of sponsored ads.
We interpret this as saying that for node $v$, there is an upper bound
\ExposureBound{v} on the number of seed sets it can be contained in, i.e.,
$\SetCard{\Set{j}{v \in \Active[j]{0}}} \leq \ExposureBound{v}$.%
\footnote{This model in fact subsumes one in which
each sponsored ad exposure only activates $v$ with some probability $p^{(j)}_{v}$:
add a new \emph{ad node} $v'$ with activation function
$\INFFUN[j]{v'} \equiv 0$,
and set $\InfFunP[j]{v}{S \cup \SET{v'}} = 1-(1-p^{(j)}_{v})(1-\InfFun[j]{v}{S})$.
Then set $\ExposureBound{v} = 0$, so only the new ``ad nodes'' can be targeted.}

In addition to the node exposure constraints \ExposureBound{v},
following the model of Datta et al., we also allow for constraints on
the seed set sizes of each advertiser:
for each advertiser $j$, the number of seed nodes $\SetCard{\Active[j]{0}}$
cannot exceed \SBUDGET[j].
In addition, we can constrain the total number of sponsored ads (more tightly):
$\sum_j |\Active[j]{0}| \leq \SBUDGET$.
In summary, our target optimization problem is the following:

\begin{definition}[Multi-Product Influence Maximization]
  For each advertiser $j$, choose a seed set \Active[j]{0} with
  $\SetCard{\Active[j]{0}} \leq \SBUDGET[j]$, such that
  $\SetCard{\Set{j}{v \in \Active[j]{0}}} \leq \ExposureBound{v}$.

   Subject to these constraints, maximize
  $\EMObj{\Active[1]{0}, \Active[2]{0}, \ldots, \Active[m]{0}}$
  or $\MEObj{\Active[1]{0}, \Active[2]{0}, \ldots, \Active[m]{0}}$.
\end{definition}

The Multi-Product Influence Maximization problem is clearly NP-hard,
as it subsumes the standard influence maximization problem.

\section{A General Result} \label{sec:general}
We begin with a very general treatment.
Under the assumptions we make
(the \VALFUN[j] are submodular functions of \FinalActive[j],
or concave functions of $\SetCard{\FinalActive[j]}$),
the overall objective functions \EMOBJ and \MEOBJ are both submodular.
In the case of \EMOBJ, this is because a non-negative linear combination
of submodular functions
(in particular: a convex combination of submodular functions)
is submodular.
In the case of \MEOBJ, the reason is that
$\SetCard{\FinalActive[j]}$ is a monotone submodular function of
\FinalActive[j], 
and applying a monotone concave function preserves submodularity.
In both cases, we can therefore apply the result of Mossel and
  Roch~\cite{mossel:roch:submodular}, which guarantees submodularity
  of the objective for each advertiser, as a function of the seed set.

To deal with the constraints on node exposures and seed set sizes,
we use a technique discussed in
\cite{calinescu:chekuri:pal:vondrak,vondrak:submodular-welfare}
in the context of finding welfare-maximizing assignments of items to
individuals, and also implicit in the work of Tang and Yuan for
  coordinating multiple advertisers \cite{tang2016optimizing}.
We create a disjoint union of separate networks
(each with the original $n$ nodes)
for each advertiser $j=1, \ldots, m$,
and then consider joint constraints on seed sets that can be selected
across all of the separate networks.
The overall new network \EXPGRAPH has $nm$ nodes,
one node $u_{v,j}$ for each combination
of an original node $v$ and advertiser $j$.
Let $\NODES[j] = \Set{u_{v,j}}{v \in V}$ denote the set of new nodes
for advertiser $j$,
$X_v := \Set{u_{v,j}}{j=1, \ldots, m}$ the set of new nodes
corresponding to the node $v$, an
and $\NODES = \bigcup_j \NODES[j] = \bigcup_v X_v$.
The influence function for node $u_{v,j}$ is
$\InfFun{u_{v,j}}{S} := \InfFun[j]{v}{\Set{v'}{u_{v',j} \in S}}$,
and thus only depends on the nodes in the network for advertiser $j$.
Writing $\FinalActive \subseteq \NODES$ for the (random) final set of
active nodes in \EXPGRAPH, the objective function is 
$\EMObj{\Active{0}} = \sum_j \ValFun[j]{\Abs{\Set{v}{u_{v,j} \in \FinalActive}}}$,
which is submodular by \looseness=-1construction.
  
Targeting the node $u_{v,j}$ in \EXPGRAPH corresponds to exposing node
$v$ to ad $j$ in the original problem.
In this way, the problem is simply to choose a subset $S \subseteq \NODES$.
The correspondence between seed sets of \EXPGRAPH and ad seeding
choices in the original problem is that
$\Active[j]{0} = \Set{v}{u_{v,j} \in \Active{0}}$.
The constraints on the selection of \Active{0} are then that
$\SetCard{\Set{j}{u_{v,j} \in \Active{0}}} \leq \ExposureBound{v}$,
$\SetCard{\Active[j]{0}} \leq \SBUDGET[j]$
(as well as $\SetCard{\Active{0}} \leq \SBUDGET$).

Because the sets \NODES[j] form a disjoint partition of \NODES,
the restriction that at most \SBUDGET[j] nodes from \NODES[j] may be
selected defines a partition matroid.%
\footnote{Readers unfamiliar with the standard definitions of
matroids and (truncated) partition matroids are directed to
Appendix~\ref{sec:matroid-definitions} or \cite{oxley:matroid-theory}.}
Similarly, because the sets $X_v$ form a partition,
the constraint that at most \ExposureBound{v} nodes from $X_v$ may be
selected defines a different partition matroid.
The constraint on the total number of selected nodes can be added to
either of these partition matroids, and turns it into a
truncated partition matroid.
The fact that the exposure constraints and an overall
constraint on the total number of seeds together define the independent sets
of a matroid, was also proven by Tang and Yuan \cite{tang2016optimizing} (Lemma~3).
%
%
%
Thus, a consequence of this reduction to the optimization problem
on \EXPGRAPH is that the constraints on \Active{0} form the
intersection of two matroids:
a partition matroid on the sets \NODES[j],
and a truncated partition matroid defined by the sets $X_v$ and the
overall constraint of \SBUDGET on the seed set size.
The optimization goal is to select a set \Active{0},
subject to these two matroid constraints, 
so as to maximize a non-negative, monotone, submodular function.
The latter is a well-studied problem.
The key results for our purposes are the \looseness=-1following:

\begin{theorem}[Theorem 3.1 of \cite{lee:sviridenko:vondrak:multiple-matroids}]
\label{thm:submodular-multiple-matroids}
There is a simple local search algorithm for maximizing a non-negative
non-decreasing submodular function subject to $k \geq 2$ matroid constraints.
For any $\epsilon > 0$, the algorithm (with suitable termination
condition) provides a polynomial-time $\frac{1}{k+\epsilon}$-approximation.
\end{theorem}

\begin{theorem}[\cite{fischer:nemhauser:wolsey}]
\label{thm:submodular-greedy}
The greedy algorithm, which always adds the next element maximizing the
increase in the objective function
(subject to not violating the matroid constraint),
is a $\frac{1}{k+1}$ approximation algorithm for maximizing a
non-negative non-decreasing submodular function
subject to $k$ matroid constraints.
\end{theorem}

Datta et al.~prove an approximation guarantee of $1/3$ for a greedy
hill climbing algorithm;
as discussed above, they do not consider budget constraints on the
total number of exposures.
Furthermore, they consider only the special case when all advertisers
have the same influence functions at all nodes.
Their proof shows that the constraints form a $p$-system for $p=2$
(a generalization of the intersection of two matroids);
they then invoke an analysis of Calinescu et
al.~\cite{calinescu:chekuri:pal:vondrak} for such systems.
The same result can be obtained by appealing to
Theorem~\ref{thm:submodular-greedy} instead. 
By making the connection to 
the intersection of matroids
explicit\footnote{Datta et al.~also discuss matroids, though
	mostly to remark that the constraints do not form a matroid, and thus
	to motivate an analysis in terms of $p$-systems.}
and invoking Theorem~\ref{thm:submodular-multiple-matroids} instead,
we improve the approximation guarantee for this problem to essentially
\half, under a much more general problem setting.%
\footnote{It should be noted that the analysis of Datta et al.~is not
	really specific to the model they formulate,
	and could easily be extended to the general problem statement from
	Section~\ref{sec:problem-statement}.}

As opposed to the model of Datta et al.,
the model of Aslay et al.~does not impose constraints 
on the seed set sizes of individual advertisers;
it only restricts exposure per node and the overall seed set size.
As discussed above, these constraints together
form a \emph{single} truncated partition matroid
(as opposed to the intersection of \emph{two} matroids).
Thus, as observed by Tang and Yuan,
the greedy algorithm guarantees a
$\frac{1}{2}$-approximation for the revenue maximization problem
(without a penalty for overshooting), 
simply by invoking Theorem~\ref{thm:submodular-greedy}
(Lemma 4 in \cite{tang2016optimizing}).


	
In fact, in this simplified model without advertiser-specific seed
set constraints,
the approximation guarantee can be improved
to $(1-1/\e)$ by appealing to the following theorem of Vondr\'{a}k et al.

\begin{theorem}[\cite{calinescu:chekuri:pal:vondrak,vondrak:submodular-welfare}]
  \label{thm:continuous-greedy}
The continuous-greedy algorithm of
\cite{calinescu:chekuri:pal:vondrak,vondrak:submodular-welfare}
is a polynomial-time $(1-1/\e-\epsilon)$-approximation
for maximizing a non-negative monotone submodular function
subject to any matroid constraint,
for arbitrarily small $\epsilon > 0$.
\end{theorem}

Unfortunately, while the continuous-greedy algorithm runs in
polynomial time, it is not practical.

The main message of this section is that, when focusing on the
arguably more natural objective of maximizing the network's revenue
(rather than minimizing regret),
non-trivial (and in some cases well-established)
algorithmic techniques can be leveraged to obtain algorithms with
approximation guarantees essentially matching those of standard
influence maximization.
These guarantees apply in a wide variety of influence models (allowing
for completely different influence functions for different advertisers),
and under a variety of constraints about the seed sets for different
advertisers.
In fact, by leveraging subsequent work on submodular maximization
under multiple Knapsack and matroid constraints (see, e.g.,
\cite{chekuri:vondrak:zenklusen,gharan:vondrak:annealing,%
lee:sviridenko:vondrak:multiple-matroids,vondrak:approximability}),
one can also obtain (somewhat weaker) approximation guarantees when
different seed nodes have different costs (in terms of money or
effort) for targeting.

\section{An LP-rounding Algorithm for the Independent Cascade Model} \label{sec:IC}
In this section, we focus on the special case of the Independent
Cascade (IC) Model for influence in a network,
and a budgeted linear valuation function:
the payoff the host receives from advertiser $j$ is
$\ValFun[j]{\FinalActive[j]}
= \min(\BUDGET[j], \PER[j] \cdot \SetCard{\FinalActive[j]})$.
The IC Model for influence has been widely studied;
a reader wishing to briefly review the definition may find
a description in Appendix~\ref{sec:IC-definition}.
In talking about the model, we use \EdgeProb[j]{e} to denote
the probability that an edge $e=(u,v)$ activates node $v$ for product
$j$ when node $u$ has become active for product $j$.

We use the idea of Reverse Reachability sets
\cite{borgs:brautbar:chayes:lucier,%
tang:shi:xiao:martingale,%
tang:xiao:shi:near-optimal}
to reduce the multi-advertiser campaign coordination problem to a
generalized Maximum Coverage problem;
we then show how to round the LP, using an algorithm of Gandhi et
al.~\cite{gandhi:khuller:parthasarathy:srinivasan:dependent},
to obtain a polynomial-time and reasonably practical
$(1-1/\e)$-approximation algorithm for the problem with both node and
advertiser constraints.

This approximation guarantee gives an improvement over the
\half-approximation guarantee of the local
search algorithm (Theorem~\ref{thm:submodular-multiple-matroids}).
When the seed set sizes of individual advertisers are not restricted, 
the continuous greedy algorithm of Theorem~\ref{thm:continuous-greedy}
matches the LP rounding algorithm's $(1-1/\e)$ guarantee,
but the continuous greedy algorithm is completely impractical.
The LP rounding  algorithm is less efficient than the simple greedy
algorithm from Theorem~\ref{thm:submodular-greedy},
but provides better approximation guarantees.

The downside, compared to the more general treatment in
Section~\ref{sec:general}, is that the results only hold for the
IC Model, and the budget constraints on the number of
exposures will only be satisfied in expectation:
in any one run, the algorithm may charge an advertiser $j$
(significantly) more than \BUDGET[j].

\subsection{The Reverse Reachability Technique}
A very useful alternative view of the Independent Cascade Model was
first shown in \cite{InfluenceSpread}, and heavily used in subsequent
work: generate graphs \GRAPH[j] by including each edge $e$ in
\GRAPH[j] independently with probability \EdgeProb[j]{e}.
Then, the distribution of nodes activated by ad $j$ in the end,
when starting from the set \Active[j]{0},
is the same as the distribution of nodes reachable from \Active[j]{0} in
the random graph \GRAPH[j].

This alternative view forms the basis of the
Reverse Reachability Set Technique, first proposed and analyzed by
Borgs, Brautbar, Chayes, and Lucier \cite{borgs:brautbar:chayes:lucier},
and further refined by Tang, Xiao, and Shi
\cite{tang:shi:xiao:martingale,tang:xiao:shi:near-optimal}.
The primary goal of 
\cite{borgs:brautbar:chayes:lucier,tang:shi:xiao:martingale,tang:xiao:shi:near-optimal}
was to permit a more efficient evaluation of the objective
\FinalSize{\Active{0}}.
Under both the Independent Cascade and Linear Threshold Models,
evaluating \FinalSize{\Active{0}} is known to be \#P-complete
\cite{chen:yuan:zhang:scalable,chen:wang:wang:prevalent}.
Reverse Reachability sets permit a more efficient
(both theoretically and practically) approximate evaluation,
as compared to the obvious Monte Carlo simulation.

We explain the Reverse Reachability technique for a single ad, and
omit the index $j$ for readability for now.
The key insight is the following:
let $v$ be an arbitrary node, and consider the (random) set $R$ of all
nodes that can reach $v$ in the randomly generated graph \GRAPH.
Then, the probability that $v$ is activated starting from
\Active{0} is equal to the probability that
$\Active{0} \cap R \neq \emptyset$.
(See, e.g., \cite[Lemma 2]{tang:xiao:shi:near-optimal}.)

More generally,
if we draw $\rho$ such sets $R_1, \ldots, R_\rho$ independently,
and for independently uniformly random target nodes $v$, 
and \Active{0} intersects \Inter{\Active{0}} of them,
then $\frac{n \cdot \Inter{\Active{0}}}{\rho}$ is an unbiased
estimator of the expected number of nodes activated 
when starting from \Active{0}.
In order to leverage this insight for computational savings in
computing ad campaigns, it is important that the estimate be
sufficiently accurate with sufficiently high probability, 
for a small enough number $\rho$ of reverse reachable sets.
A characterization of this accuracy is given by the following lemma
from \cite{tang:xiao:shi:near-optimal} (restated slightly here):

\begin{lemma}[Lemma 3 of \cite{tang:xiao:shi:near-optimal}]
\label{lem:reverse-reachability}
Assume that
\begin{align}
  \rho & \geq (8+2\epsilon) \cdot n \cdot
         \frac{c \log n + \log {n \choose k} + \log 2}{\OPT \cdot \epsilon^2}.
\end{align}
Consider any set $S$ of at most $k$ nodes.
With probability at least $1 - n^{-c}/{n \choose k}$,
the following inequality holds for $S$:
\begin{align}
\left|   \frac{n \cdot \Inter{S}}{\rho}
       - \FinalSize{S} \right|
& < \frac{\epsilon}{2} \cdot \OPT.
\label{eqn:rr-approximation}
\end{align}
Here $\OPT$ is the maximum influence spread that can be achieved by
any set $S$ of at most $k$ nodes.
\end{lemma}

Lemma~\ref{lem:reverse-reachability} implies that 
when $\rho$ is large enough to allow taking a union bound over all
relevant sets $S$, the fraction of sets $R_i$ that intersect a
candidate seed set $S$ is an accurate stand-in for the actual
objective function \FinalSize{S}.
As already observed by Borgs et al.~\cite{borgs:brautbar:chayes:lucier},
this reduces the problem of influence maximization to a Maximum
Coverage problem:
selecting a set of at most $k$ seed nodes that jointly maximize the
number of sets $R_i$ containing at least one seed node.

Without loss of generality, assume that $\SBUDGET[j] \leq \SBUDGET$
for all $j$.
Let \OPT[j] be the maximum influence that advertiser $j$ could achieve
with \SBUDGET[j] seed nodes, and ignoring the budget constraint \BUDGET[j].
A lower bound \OPTH[j] on each \OPT[j] can be found by running the TIM
algorithm \cite{tang:xiao:shi:near-optimal}
with seed sets of size \SBUDGET[j]. 
For each advertiser $j$, we draw
$\rho^{(j)} \geq 9 n^3 m^2 \cdot
\frac{c \log n + \log m + \log {n \choose \SBUDGET[j]} + \log 2}{%
      \OPTH[j] \cdot \epsilon^2}$
nodes%
\footnote{This value of $\rho^{(j)}$ is chosen with foresight
    to apply Lemma~\ref{lem:reverse-reachability} and allow the
    application of various bounds in the following paragraphs.
    }
$v$ i.i.d.~uniformly (in particular, with replacement) from \NODES[j],
and for each such $v$,
we let $R \subseteq \NODES[j]$ be the random reverse reachable set.
For later ease of notation,
we index the nodes as $v_i \in \NODES$, and the corresponding sets as $R_i$.
Notice that the same node $v$ might appear as $v_i$ for different $i$,
with possibly different sets $R_i$.
For each such index $i$ of a node $v_i$ and set $R_i$,
let $j(i)$ be the unique advertiser $j$
such that $R_i \subseteq \NODES[j]$,
and write $\RRP[j] = \Set{i}{j(i) = j}$ for the set of all sets used
to estimate the influences for advertiser $j$. 

Focus on any advertiser $j$, and the influence of any set
$S \subseteq \NODES[j]$ with $\SetCard{S} \leq \SBUDGET[j]$.
Using Lemma~\ref{lem:reverse-reachability} with
$\epsilon' = \frac{\epsilon}{nm} \leq \half$ and $c'=c+\log_n(m)$,
the influence of $S$ is estimated to within an additive term
$\frac{\epsilon}{nm} \cdot \OPT[j]$ with probability at least 
$1 - n^{-c'}/{n \choose \SBUDGET[j]}
= 1 - 1/(n^c m) \cdot 1/{n \choose \SBUDGET[j]}$.
By a union bound over all such sets $S$ (of which there are%
\footnote{When running a greedy algorithm, only $O(mn\SBUDGET[])$ such sets
  need to be considered, which leads to significant computational
  savings. Here, our goal is to leverage the reverse reachability
  technique not necessarily for significant computational savings, but
  for improved approximation guarantees.}
${n \choose \SBUDGET[j]}$) and all $m$ advertisers $j=1, \ldots, m$,
the influence of \emph{all} such sets $S$ is simultaneously estimated
to within an additive term of $\frac{\epsilon}{nm} \cdot \OPT[j]$,
with probability at least $1 - 1/n^c$.
Assume for the rest of this section that the high-probability event
has happened.
 
Now consider an arbitrary seed set $S$ with $\SetCard{S} \leq \SBUDGET$,
not necessarily contained in just one partition \NODES[j].
Let $S^{(j)} = S \cap \NODES[j]$, and let 
$q^{(j)} = \frac{n \cdot \SetCard{\Set{i \in \RRP[j]}{R_i \cap S^{(j)} \neq \emptyset}}}{\rho}$
be the estimated expected influence of $S^{(j)}$.
By the preceding paragraph, the influence of each $S^{(j)}$ is estimated
to within an additive $\frac{\epsilon}{nm} \cdot \OPT[j]$, i.e.,
$\Abs{q^{(j)} - \FinalSize{S^{(j)}}} \leq \frac{\epsilon}{nm} \cdot \OPT[j]$.
We next want to show that for each $j$, the actual payoff the host
obtains from advertiser $j$, i.e.,
$\min(\BUDGET[j], \PER[j] \cdot \FinalSize{S^{(j)}})$ is estimated to
within an additive term $\epsilon/m \cdot \OPT$;
then, by summing over all $m$ advertisers $j$, the total estimation error
is at most $\epsilon \cdot \OPT$.

To prove the estimation error bound for an advertiser $j$,
we consider several cases.
First, if $S^{(j)} = \emptyset$, then the payoff 0 is estimated
completely accurately.
Second, if $\BUDGET[j] \leq \PER[j]$, i.e., the advertiser's budget is
so low that he pays for at most one impression, the payoff is also
estimated completely accurately.
The reason is that for any $S^{(j)} \neq \emptyset$, all nodes in $S^{(j)}$
will always be shown impressions,
so $\PER[j] \cdot q^{(j)} \geq \BUDGET[j]$ with probability 1.
As a result, the payoff from the advertiser is estimated correctly as
\BUDGET[j].
Finally, if $\BUDGET[j] \geq \PER[j]$, then
$\OPT[j] \leq n \cdot \OPT/\PER[j]$.
The reason is that the optimum does at least as well as using \SBUDGET[j]
seed nodes for advertiser $j$, which would yield payoff
$\OPT \geq \min(\BUDGET[j], \PER[j] \cdot \OPT[j])
\geq \PER[j] \cdot \min(1, \OPT[j])
\geq \PER[j] \OPT[j]/n$, because $\OPT[j] \leq n$.
Thus, the estimation error for the payoff from advertiser $j$ is at
most 
$\PER[j] \cdot \frac{\epsilon}{nm} \cdot \OPT[j]
\leq \frac{\epsilon}{m} \cdot \OPT$.%
\footnote{
Some of the complications in the proof arose because of the
possibility of very different budgets and payoffs per node across
advertisers.
When the payoff per node for all advertisers is equal,
i.e., $\PER[j] = \hat{\PER}$ for all $j$,
then
$\rho^{(j)} \geq (8+2\epsilon)n^2m^2\hat{\PER}^{2}\cdot \frac{c\log {n} + 
\log {n \choose K} + \log {K} + \log {2}}{\epsilon^2 OPT^2}$
RR set samples are enough to provide accurate payoff estimates with
sufficiently high probability. 
}

\subsection{The LP and Rounding Algorithm}\label{sec:LPRound}
Notice that except for the budget caps \BUDGET[j],
the objective is a sum of (scaled) coverage functions.
This suggests the natural formulation of an LP, which allows us to 
use a randomized rounding technique due to Gandhi et
al.~\cite{gandhi:khuller:parthasarathy:srinivasan:dependent}.

Recall that each set $R_i$ is generated by a reverse (random) BFS from
a randomly sampled node $v \in \NODES$, where the same node $v$ may be sampled
for multiple $i$.
Also recall that $j(i)$ is the unique advertiser for the reverse
reachable set $R_i$,
and that \RRP[j] is the set of all 
reverse reachability sets $R_i$ with $j(i) = j$.
We use the decision variable \Select[j]{v} to denote whether node $v$
was targeted for ad $j$,
\SSel{i} for the decision variable whether at least one node from the set
$R_i$ was selected (for its corresponding ad), 
and the variable \Rev[j] for the total revenue from advertiser $j$.
The optimization objective can be expressed using the following
Integer Linear Program (ILP):

\begin{LP}[eqn:ILP]{Maximize}{\sum_{j} \Rev[j]}
  \SSel{i} \leq \sum_{v \in R_i} \Select[j(i)]{v} & \mbox{ for all } i \\
  \SSel{i} \leq 1 & \mbox{ for all } i\\
  \sum_j \Select[j]{v} \leq \ExposureBound{v} & \mbox{ for all } v\\ 
  \sum_v \Select[j]{v} \leq \SBUDGET[j] & \mbox{ for all } j\\ 
  \sum_{j,v} \Select[j]{v} \leq \SBUDGET\\
  \Rev[j] \leq \frac{n}{\rho}\cdot \sum_{i \in \RRP[j]} \PER[j] \SSel{i} & \mbox{ for all } j\\
  \Rev[j] \leq \BUDGET[j] & \mbox{ for all } j\\
  \SSel{i}, \Rev[j], \Select[j]{v} \geq 0 & \mbox{ for all } i, j, v\\
  \Select[j]{v} \in \SET{0,1} & \mbox{ for all } j, v.
\end{LP}

The first constraint states that a target node $v_i$ is only covered
if at least one node from the set $R_i$ is selected,
and the second constraint ensures that the LP derives no benefit from
double-covering nodes.
The third constraint encodes the bound on the number of targeted ads
that $v$ can be exposed to,
while the fourth constraint encodes the bound on the seed set sizes of
the advertisers $j$.
The fifth constraint captures the overall bound on the total number of
targeted ads.
The sixth and seventh constraints characterize the objective function
for advertiser $j$.

Of course, solving the ILP is NP-hard, so as usual, we consider the
fractional LP relaxation obtained by omitting the final integrality
constraint $\Select[j]{v} \in \SET{0,1}$.
Without this constraint, the LP can be solved in polynomial time,
obtaining a fractional solution $(\SSelV, \RevV, \SelectV)$.

To round the fractional solution,
we use the dependent rounding algorithm of Gandhi et al.
(See Section~2 of \cite{gandhi:khuller:parthasarathy:srinivasan:dependent}.)
This algorithm can be applied to round fractional
solutions whenever the variables to be rounded (here: the
\Select[j]{v}) can be considered as edges in a bipartite graph,
and constraints correspond to bounds on the (fractional) degrees of the
nodes of the bipartite graph.\footnote{%
The fifth constraint of our LP (overall budget) does not fit in
this framework, since it is a constraint on all variables.
We extend the rounding algorithm of Gandhi et al.~by one more final
loop (see below) which repeatedly picks any two fractional variables
and randomly rounds them using the same approach as the rest of the algorithm, 
until only one fractional variable remains.
This last variable can then be rounded randomly by itself.}
It repeatedly rounds the fractional values associated with a subset of
edges forming a maximal path or a cycle,
until all variables are integral;
let \ISelect[j]{v} denote the rounded (integral) values corresponding
to \Select[j]{v}.
The algorithm of \cite{gandhi:khuller:parthasarathy:srinivasan:dependent}
provides the following guarantees:
\begin{itemize}
\item The probability that the variable \Select[j]{v} is rounded to
  $\ISelect[j]{v} = 1$
  is exactly \Select[j]{v}, for all $j$ and $v$.
\item In the bipartite graph, each node's degree is rounded up or down
to an adjacent integer.
In the context of our LP~\eqref{eqn:ILP},
this means that the number of seed nodes for each advertiser is either
rounded up or down to the nearest integer,
as is the number of exposures to sponsored ads for each node.
In particular, the third and fourth constraints are satisfied with
probability 1.
\item For any node of the bipartite graph and any subset of incident
edges, the rounded values are negatively correlated.
In the context of our LP~\eqref{eqn:ILP},
for any advertiser $j = j(i)$ and reverse reachable set $R_i$, 
the algorithm guarantees that
$\Prob{\ISelect[j]{v} = 0 \text{ for all } v \in R_i}
\leq \prod_{v \in R_i} \Prob{\ISelect[j]{v} = 0}
= \prod_{v \in R_i} (1-\Select[j]{v})$.
\end{itemize}

After rounding all the \Select[j]{v} to \ISelect[j]{v},
for each $i$, we set
$\ISSel{i} = 1$ if at least one $v \in R_i$ has $\ISelect[j(i)]{v} = 1$,
and for each $j$, we set
$\IRev[j] = \sum_{i \in \RRP[j]} \PER[j] \cdot \ISSel{i}$.
Together with the guarantees of Gandhi et al., this ensures that 
all ILP constraints except possibly the seventh (budget) constraint
are satisfied; the budget constraint is discussed more in
Section~\ref{sec:budget-constraints}.


\begin{theorem} \label{thm:lp-rounding-guarantee}
Under the IC model, with constraints on node exposures and advertiser
seed set sizes, 
the (modified) Gandhi et al.~correlated LP-rounding based
algorithm is a polynomial-time $(1-1/\e)$-approximation algorithm,
if advertisers pay based on their \emph{expected} exposure.
\end{theorem}
\begin{proof}
Consider any set $R_i \in \RRP[j]$.
With explanations of some steps given below,
the probability that $\ISSel{i} = 1$ is
\begin{align*}
\Prob{\ISSel{i} = 1}
    & = 1 - \Prob{\ISelect[j(i)]{v} = 0 \text{ for all } v \in R_i}
\\  & \stackrel{(1)}{\geq} 1 - \prod_{v \in R_i} \Prob{\ISelect[j(i)]{v} = 0}
\\  & = 1 - \prod_{v \in R_i} (1-\Select[j(i)]{v})
\\  & \stackrel{(2)}{\geq} 1 - \e^{-\sum_{v \in R_i} \Select[j(i)]{v}}
\\  & \geq 1 - \e^{-\min(1, \sum_{v \in R_i} \Select[j(i)]{v})}
\\  & \stackrel{(3)}{\geq} (1-1/\e) \cdot \min(1, \sum_{v \in R_i} \Select[j(i)]{v})
\\  & \stackrel{(4)}{\geq} (1-1/\e) \cdot \SSel{i}.
\end{align*}

The inequality labeled (1) used the negative dependence discussed above.
The inequality labeled (2) used the standard bound that $1-x \leq
\e^{-x}$ for all $x \geq 0$.
The inequality labeled (3) uses that the function
$x \mapsto 1-\e^{-x}$ is concave,
and hence on the interval $[0,1]$ is lower-bounded by the straight
line which agrees with it at the endpoints $x=0$ and $x=1$,
which is the function $x \mapsto (1-1/\e) \cdot x$.
Finally, the inequality labeled (4) uses that
$\SSel{i} \leq \min(1, \sum_{v \in R_i} \Select[j(i)]{v})$
by the first and second LP constraints.
Using the definition of \IRev[j], we now get that

\begin{align*}
\Expect{\IRev[j]}
& = \Expect{\sum_{i \in \RRP[j]} \PER[j] \cdot \ISSel{i}}
\; = \; \sum_{i \in \RRP[j]} \PER[j] \cdot \Prob{\ISSel{i} = 1}
\\ & \geq \sum_{i \in \RRP[j]} \PER[j] \cdot (1-1/\e) \cdot \SSel{i}
\; \geq \; (1-1/\e) \cdot \Rev[j].       
\end{align*}

\end{proof}

\subsection{No Seed Set Size Restrictions}
\label{sec:unrestricted-seed-sets}

In the model of Aslay et al.,
and in our experiments in Section~\ref{sec:experiments}, 
there are no restrictions placed on the seed set sizes of individual
advertisers;
that is, the fourth set of constraints in the LP~\eqref{eqn:ILP}
is missing. 
In the absence of these constraints,
the correlated LP-rounding algorithm of Gandhi et
al.~\cite{gandhi:khuller:parthasarathy:srinivasan:dependent} can be
further sped up.
Note that except for the fifth (sum of seed set sizes) constraint,
the \Select[j]{v} constraints are a collection of disjoint stars.
As a result, the general correlated rounding algorithm of Gandhi et
al.~can be replaced with their algorithm for stars.
The resulting algorithm runs in linear time and is very simple to implement.
It first uses the correlated rounding on pairs of variables
\Select[j]{v}, \Select[j']{v} corresponding to the same 
node (and different advertisers) until no such pairs remain;
subsequently, as for the general algorithm,
we add one more \textbf{while} loop to round pairs of the remaining
\Select[j]{v} (which now necessarily correspond to distinct nodes $v,v'$).

Because there is at most one fractional variable remaining incident on
any node $j$ or $v$, the second property of the Gandhi et al.~rounding
algorithm is preserved.
The negative correlation property holds inductively,
because at most two values are rounded in each step.

\subsection{The Budget Constraints}
\label{sec:budget-constraints}

In Theorem~\ref{thm:lp-rounding-guarantee}, we emphasized that
advertiser payments have to be based on the expected exposure,
rather than the actual exposure.
Indeed, using the LP rounding approach, maintaining the budget
constraints in each run is not possible without a huge ($\Omega(m)$)
loss in the approximation guarantee.
The reason is that LP~\eqref{eqn:ILP} has an integrality gap of
$\Omega(m)$ when the budgets are constrained.
Consider the following input: the influence graph is a star with $n$ leaves,
in which the center node has probability 1 of influencing all its
neighbors, and no one else can influence any other node.
Thus, any advertiser who gets to advertise to the center node will
reach the entire graph.

Suppose that each advertiser $j$ has a budget of
$\BUDGET[j] = n/m$, and the center node $v$ has a constraint of
$\ExposureBound{v} = 1$.
A fractional solution can allocate a fraction of $1/m$ of the center
node to each advertiser; each advertiser then reaches $n/m$ nodes,
not exceeding his budget, and the host obtains a payoff of $n$.
But any integral solution can only allocate $v$ to one advertiser,
who will then pay at most his budget $n/m$.
Thus, any integral solution obtains payoff at most $n/m$.

This problem disappears if each advertiser's budget must only be met
in expectation.
In other words, on any given day, the advertiser's budget could be
exceeded, even by a lot.
But on average, the advertiser does not pay more for impressions than
\BUDGET[j] per day.
Because the number of impressions resulting from a particular rounding
choice is a random variable in $[0,n]$,
standard tail bounds (e.g., Hoeffding Bounds) show that the total
payment is with high probability close to the expected payment after
about $O(n)$ days.

\section{Experiments} \label{sec:experiments}
\newcommand{\RP}{\ensuremath{\alpha}\xspace}

In this section, we 
describe an empirical evaluation of
the relationship between the total revenue to the host and the algorithm used,
the number of advertisers, and the total number of sponsored ads.
More importantly, we also consider the interplay of 
competition with the influence strengths in the network and 
the similarity/dissimilarity of the advertisers' influence networks.
We conducted our experiments on a 36-core Linux server with 
an Intel Xeon E5-2695 v4 processor at 2.1GHz and 1TB memory. 
We develop a parallel version of the Greedy algorithm that scales
to networks of millions of nodes and tens of millions of edges,
with dozens of advertisers (see Section \ref{sec:scalability}).


All of our experiments are based on the Independent Cascade (IC) Model.
We employed the Reverse Reachability technique to (approximately)
compute the influence of node sets more efficiently.
In Section \ref{sec:IC}, we obtained an upper bound on the number of
RR sets $\rho^{(j)}$ to guarantee good approximations with high probability.
This bound is quite large, and using the
corresponding number of sets would not permit us to scale experiments
to large networks.
We conducted extensive experiments on the number of RR sets that are
sufficient to guarantee high accuracy \emph{in practice}.
These experiments are described in detail in
Appendix~\ref{sec:RR-set-experiments};
the upshot is that $\rho^{(j)} \geq 10 n$ RR sets are sufficient
to guarantee an estimation error $\epsilon \leq 2\%$.
Therefore, for the experiments reported in the remainder of this
section, we used a value of $\rho^{(j)} = 10 n$.

\subsection{Data Sets and Generation of IC Instances}
We used the following four real-world networks,
whose key statistics are summarized in Table~\ref{table:datasets}.
(Two larger networks used for scalability experiments are
described in Section \ref{sec:scalability}.)

\begin{itemize}
\item \Facebook \cite{konect} is an ego network of a user in
  Facebook, excluding the ego node.
\item \Advogato \cite{konect} is a social network whose nodes are
  users of the Advogato platform;
  directed edges $(u,v)$ are trust links.
\item \DBLP \cite{konect} is a citation network in which nodes
  are papers, and there is a directed edge from $u$ to $v$
  if $u$ cites $v$.
\item \NetHEPT \cite{chen2009efficient} is a collaboration network
  generated from co-authorships in high-energy physics publications.
\end{itemize}


\begin{table}[htb]
\centering
\caption{Data Set characteristics}
\label{table:datasets}
\begin{tabular}{|l|l|l|l|}
\hline
Dataset   & \#nodes & \#edges & Type       \\ \hline
\Facebook & 2,889    & 2,981    & Undirected \\ \hline
\Advogato & 6,542    & 51,127   & Directed   \\ \hline
\DBLP     & 12,592   & 49,743   & Directed   \\ \hline
\NetHEPT  & 15,229   & 31,376   & Undirected \\ \hline
\end{tabular}
\vspace{-2mm}
\end{table}

In creating influence networks (with edge probabilities \EdgeProb[j]{u,v}),
for most of our experiments,
we wanted to avoid the assumption of uniform probabilities on edges
commonly made when evaluating algorithms for influence maximization.
The reason is that one key aspect of our work is the notion of
coordination/competition between different advertisers.
When all edges have uniform probabilities \EdgeProb[j]{u,v}, the value
of a node is solely determined by its network position.
We therefore define a way of generating edge probabilities
non-uniformly that gives some nodes intrinsically more importance.
For each node $v$, we draw a parameter $\lambda_v$
independently and uniformly from $[0, 0.4]$ for \Facebook, \DBLP and \NetHEPT, 
and $[0, 0.3]$ for \Advogato.
(The different intervals were chosen to counteract the 
effects of varying edge densities in the datasets.)
We then define $\EdgeProb{u,v} = \lambda_u \cdot \lambda_v$.
We will discuss below how the \EdgeProb[j]{u,v} are correlated for
different $j$.

For all experiments, we set the constraint on the number of sponsored
ads that can be shown to node $v$ to $\ExposureBound{v}=1$,
i.e., each node can be shown a sponsored ad from at most one advertiser.
This results in maximal inter-advertiser constraints for the
seed sets.
We also performed experiments varying the value of \ExposureBound{v}.
The outcome of these experiments is that what matters is mostly
the \emph{ratio} $\ExposureBound{v}/m$.
As \ExposureBound{v} grows (with $m$ fixed),
the optimization problem gradually decouples across advertisers.
As a result, the most interesting and novel aspects of the model
manifest themselves when $\ExposureBound{v}=1$.

For simplicity, we set each advertiser's payoff per exposure to $\PER[j]=1$.
The overall numbers of sponsored ads \SBUDGET are varied for
different experiments, and discussed below.
In our experiments, we did not consider individual advertiser seed set
constraints \SBUDGET[j];
such constraints are heavily studied in traditional Influence
  Maximization experiments, and we wanted to focus on the novel
  aspects arising due to advertiser competition.

\subsection{Comparison between the Algorithms} \label{sec:compAlgo}
Our first set of experiments simply compared the performance of our
algorithms and several baseline heuristics.
We implemented the following algorithms:

\begin{itemize}[leftmargin=*]
\item \Greedy is the standard greedy algorithm for maximizing a
  submodular function subject to a matroid constraint 
  (Theorem~\ref{thm:submodular-greedy}).
\item \LPRound is described in Section~\ref{sec:LPRound}.
\item \MaxDeg considers nodes by non-increasing degrees, and assigns
  them to advertisers in a round-robin order.
  Each node $v$ is assigned to $r_v$ consecutive advertisers.
\item \EigenCent considers nodes by non-increasing eigenvector
  centrality%
\footnote{The eigenvector entries of the leading eigenvector of the
graph's adjacency matrix.},
and assigns them to advertisers in a round-robin order, as with \MaxDeg.
\item \LPOPT is the value of the optimal fractional solution of the
  LP~\eqref{eqn:ILP}. It provides an upper bound on the value of the
  optimal solution, and thus gives us a benchmark to compare the
  algorithms' performance to, on an absolute scale.
\end{itemize}

\subsubsection{Varying the total number of seeds}\label{sec:PayoffVsK}

In the first set of experiments, we kept the number of advertisers
constant at $m=3$, and varied the total number \SBUDGET of seed nodes.
Thus, these experiments are similar to evaluations of standard
Influence Maximization algorithms.
To avoid strong effects of competition between
advertisers (which we are evaluating in later sections),
we generated the $\lambda^{(j)}_v$ (and hence the \EdgeProb[j]{u,v})
independently for each $j$.

\begin{figure*}[htb]
\includegraphics[width=\textwidth]{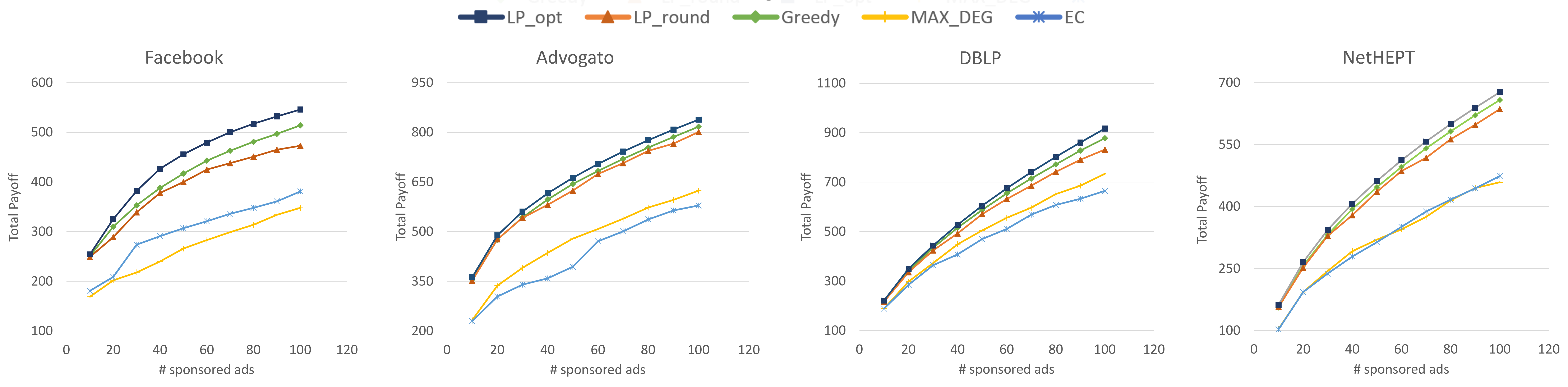}
\caption{Host payoff vs.~number of sponsored ads (\SBUDGET)
  }
  \label{fig:payoffVsK}
\end{figure*}

Figure \ref{fig:payoffVsK} shows a comparison of the total host payoff
achieved by the algorithms as \SBUDGET is varied from $10$ to $100$.
Both \LPRound and \Greedy perform significantly better than
\MaxDeg or \EigenCent.
This is not surprising, as the heuristics only consider the network
\emph{structure}, but not the influence probabilities associated with
edges.
However, the random generation of edge probabilities still ensures
that nodes of high degree or high centrality tend to be more influential;
hence, in some scenarios (especially with small \SBUDGET),
the payoffs of the \MaxDeg and \EigenCent heuristics are comparable to
those of \Greedy and \LPRound.

A comparison to the fractional LP solution value shows that both
\Greedy and \LPRound achieve more than 85\% of the optimal payoff,
which is significantly more than the respective guarantees of
\half and $1-\frac{1}{\e}$.
Experimentally, on these instances, 
\Greedy performed marginally better than \LPRound.

\subsubsection{Varying the number of advertisers}

For the second set of experiments, we varied the number $m$ of
advertisers from 1 to 20, scaling the number of seeds as
$\SBUDGET = 10m$, to keep an average of 10 seeds per advertiser.
Again, we did not add constraints on the budgets or the individual
advertisers' seed sets.

Using this set of experiments, we also studied the effect (on payoff) of
competition between advertisers, as caused by the constraint that each
node can only be chosen as a seed for one advertiser.%
\footnote{If there were no constraint on the number of sponsored ads per node,
	there would be no competition, and the total payoff would increase
	linearly in $m$.
	Note also that the total influence of the seed sets of all advertisers
	is different from the influence of their seed sets' union if assigned
	to one advertiser: while no two advertisers can choose the same seed
	node, if the different seed nodes \emph{influence} the same nodes
	later, they will both derive utility (and the host revenue) from
	those exposures.}
Therefore, we made the influence probabilities \EdgeProb[j]{u,v} the
same for all $j$, i.e., the same nodes are most influential for all
advertisers.

Figure~\ref{fig:payoffVsm} shows the per-advertiser payoff achieved 
as $m$ is varied from 1 to 20.
A comparison of the algorithms' performance yields results similar to
those reported in Section~\ref{sec:PayoffVsK}. 
For a single advertiser, the performance of the \MaxDeg heuristic is
comparable to \LPRound and \Greedy.
As explained earlier, this observation can be attributed
to the random generation of uniform edge probabilities. 
For larger values of $m$, there is a significant difference in the
performances of the \Greedy and \LPRound algorithms  
vs.~the \MaxDeg and \EigenCent heuristics.
A likely explanation is the following:
when many advertisers compete on the same network,
\Greedy and \LPRound can alleviate competition for a limited pool of
highly influential nodes by indirectly influencing such nodes using
other carefully chosen seeds.
In contrast, \MaxDeg and \EigenCent do not consider such potential
propagation for seed selection. 

\begin{figure*}
\includegraphics[width=\textwidth]{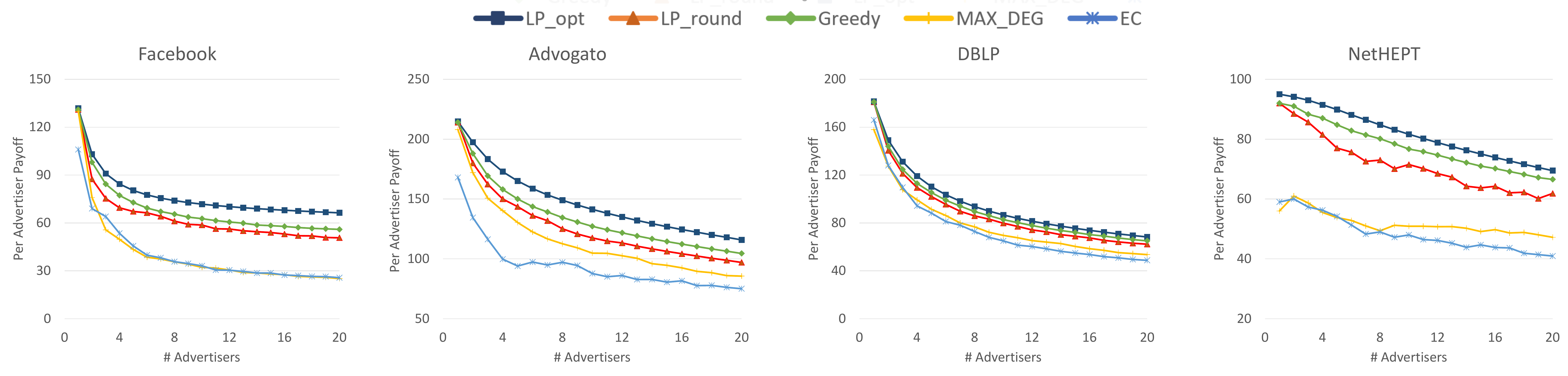}
\caption{Average per-advertiser payoff vs.~number of advertisers ($m$) }
\label{fig:payoffVsm}
\end{figure*}
      
Another observation is that as a result of the competition,
per-advertiser payoff decreases and hence,
the total host payoff does not scale linearly in $m$.
For a single advertiser, \DBLP and \Advogato have significantly higher payoff
than \Facebook and \NetHEPT.
With 20 advertisers, on the \DBLP and \Facebook networks,
the average payoff per advertiser decreases by factors of 2.8 and 2.4,
respectively, while the decrease for \Advogato is a factor of 2.
Comparatively, \NetHEPT exhibits only modest competition,
and its per-advertiser payoff decreases only by a factor of $1.4$.
A possible cause for the higher decreases may be the very skewed degree
distributions of the \DBLP and \Facebook networks,
which results in a smaller number of extremely valuable nodes,
and the fact that \Advogato and \DBLP are directed, resulting in less
symmetry between nodes.
The strong competition for \DBLP and \Facebook can also be observed in
the fact that the gradient of the per-advertiser payoff is very
negative for small values of $m$.
For larger $m$, the decrease becomes less pronounced,
because there are many marginally useful seed nodes to choose from.

\subsection{Effects of Competition}\label{sec:EOComp}
Next, we focused in detail on the effects of competition on
total host payoff.
More specifically, the question we were interested in is the following:
how does the similarity or dissimilarity of influence networks for
different advertisers affect revenue?
If the influence networks are very similar, then high-value seed nodes
for one advertiser will typically also be high-value for others,
and the constraint that no node must be chosen by more than a given
number of advertisers (in our experiments: $1$) will constrain the
reach of seed sets.
On the other hand, if the influence networks are very different,
then different advertisers might focus on different parts of the
network, and the host could potentially derive significantly higher
payoff.

Since the purpose of these experiments is not to compare the
performance of algorithms, but to draw qualitative insights,
we ran these experiments only using the \Greedy algorithm,
which had performed best in our earlier experiments.

For these experiments, we fixed the number of advertisers to $m=20$
and total seeds to $\SBUDGET = 200$.
To cover a spectrum of different similarities between networks,
we used the following generative model.
All advertisers initially have the same
edge probabilities \EdgeProb{u,v}.
We can assume for simplicity that the graph is complete, by setting
$\EdgeProb{u,v} = 0$ whenever $(u,v)$ is not an edge.
A parameter $s \in \SET{0, 1, \ldots, 200}$
will capture the similarity between the networks for different
advertisers.
Each advertiser $j$'s influence strengths \EdgeProb[j]{u,v} are
generated independently as follows.
Starting from $\EdgeProb{u,v}$, perform $s \cdot \frac{n}{100}$
\emph{node swaps} of the following form:
select two vertices $u,v$ independently and uniformly at random,
and switch all their associated influence probabilities, i.e.,
set $\EdgeProbP[j]{u,w} = \EdgeProb[j]{v,w}$ and
$\EdgeProbP[j]{v,w} = \EdgeProb[j]{u,w}$ for all $w \neq u,v$,
and $\EdgeProbP[j]{u,v} = \EdgeProb[j]{v,u}$
and $\EdgeProbP[j]{v,u} = \EdgeProb[j]{u,v}$.

The effect is that the influence networks for all advertisers $j$ are
exactly isomorphic to each other, i.e., no advertiser has an a priori
better network.
However, the larger the value of $s$, the more independent the
networks are, which we expected to lead to more potential to derive
payoff from all advertisers simultaneously.

\begin{figure}[htbp]
\centering
\includegraphics[width=0.6\linewidth]{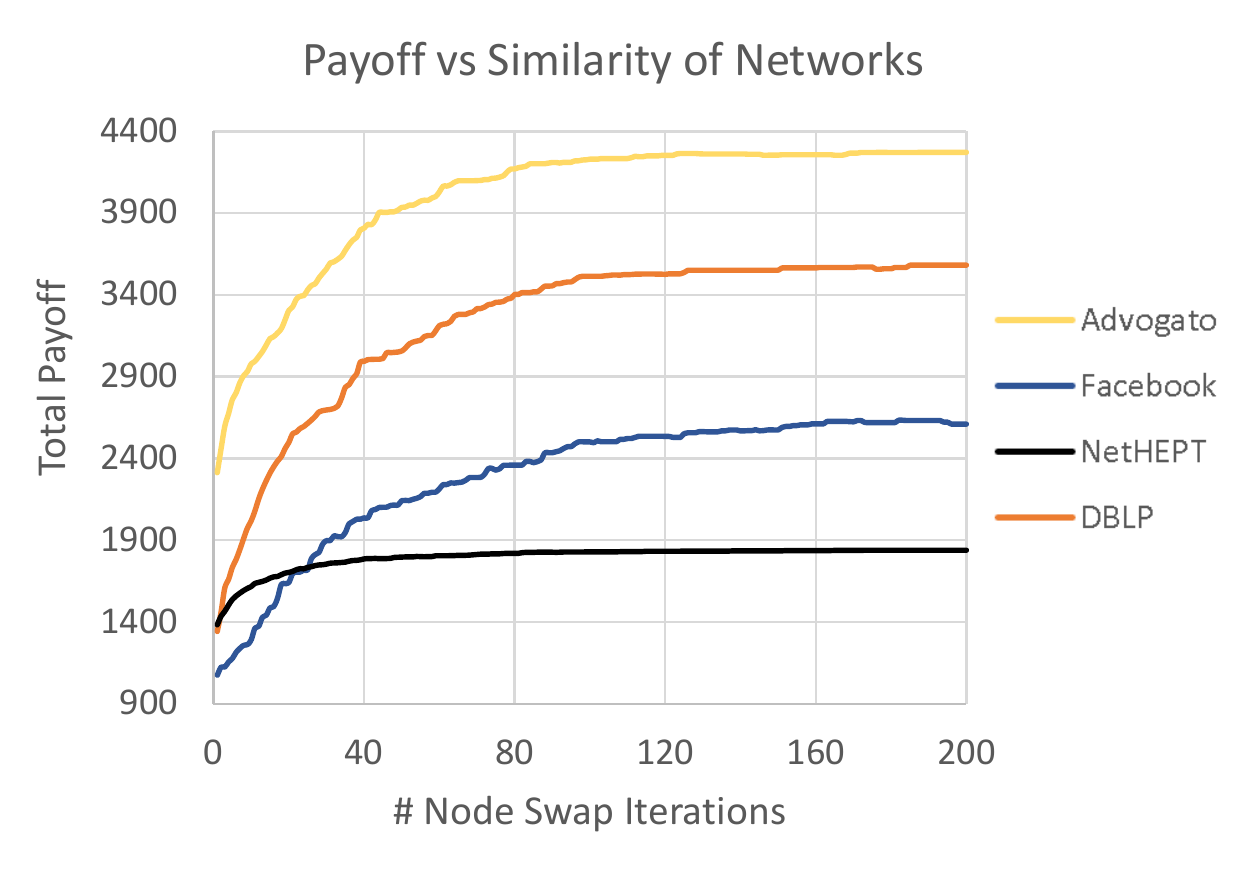}
\caption{Host payoff vs.~influence network similarity.}
\label{fig:comp}
\end{figure}


Figure~\ref{fig:comp} shows the payoff as a function of $s$.
First, notice that the host's payoff does indeed increase steeply in
$s$, nearly linearly for a non-trivial segment.
This shows that competition between advertisers for high-impact nodes
indeed restricts the host's payoff;
as the advertisers' influence networks become more dissimilar,
the host can extract more payoff from the joint ad campaigns.

Next, we observe that the relative increase in payoff
(comparing no swaps vs.~$2n$ swaps)
is noticeably higher in the \Facebook (a factor of 2.4) and
\DBLP (a factor of 2.7) networks,
compared to the \Advogato (a factor of 1.8) and \NetHEPT (a factor of 1.3)
networks.
This aligns with our earlier observations:
the \Facebook and \DBLP data sets seem to have fewer high-influence nodes,
as compared to the more even influence of nodes in \Advogato and \NetHEPT.
Thus, ensuring more independence among the isomorphic copies of the
\Facebook and \DBLP graphs creates more potential for additional \looseness=-1payoff.

Further, the payoff saturates around just $n/3$ swaps for \NetHEPT,
as opposed to around $n$ swaps for \Facebook, \Advogato and \DBLP.
This is likely because networks with more competition require
more swaps to completely realize the potential of essentially
independent \looseness=-1campaigns.

\subsection{Effect of Influence Probabilities}
For our final set of experiments, we were interested in the interplay
between competition and edge strengths.
We expected two counter-acting effects: as the probabilities on edges
increase, more different seed nodes may become capable of reaching the
same large part of the network, thus reducing the negative effects of
competition.
On the other hand, as the edge probabilities decrease,
most cascades will not spread beyond a few nodes; as a result,
all parts of the network may provide small influence,
so again, the additional detrimental effects of competition could be
reduced.

For $m$ advertisers,
we chose a combined seed set size of $\SBUDGET = 10m$,
and gave each advertiser a budget of $\BUDGET[j] = n/5$.
Different from the earlier experiments, 
we assigned \emph{uniform} probabilities of $p$ to the edges,
and varied the value of $p$.
Again, because our goal was to study the interplay between competition
and parameters of the model (rather than comparing algorithms),
we only used the \Greedy algorithm.
We were interested in the host's payoff increase as the number of
advertisers is increased from 1 to 20.
We call the ratio between the two quantities the
\emph{relative payoff at $m=20$}, and denote it by \RP.

\begin{figure}
\centering
\includegraphics[width=0.6\linewidth]{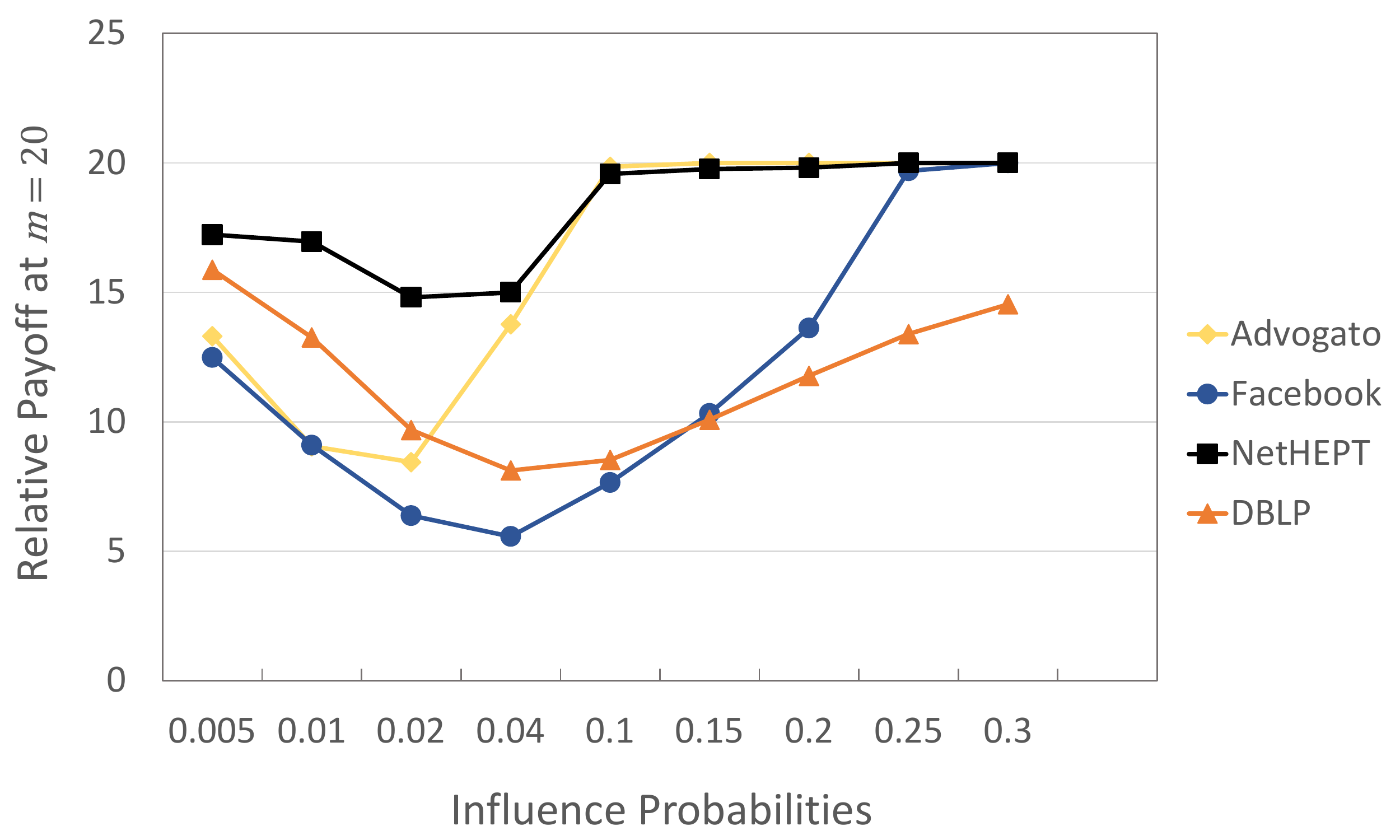}
\caption{Relative payoff with $m=20$ advertisers vs.~influence
  probability $p$ (shown in non-linear scale to display interesting
  regimes) }
  \label{fig:RPVsProb}
\end{figure}
      
Figure~\ref{fig:RPVsProb} shows the relative payoff \RP as a function
of $p$.
The two counter-acting effects produce --- for all four networks ---
a local minimum in \RP.
As $p$ grows large, the ratio saturates at 20;
this is not surprising, as with high enough probabilities, essentially
any node will reach the entire graph, and there is in effect no
competition, resulting in a 20-fold increase in host payoff.

For \Advogato and \NetHEPT (the two graphs which earlier showed less
susceptibility to competition),
the ratio \RP is minimized at the same $p=0.02$.
However, the underlying reasons appear to be different.
\Advogato has a high edge density; as a result,
advertisements spread to a large portion of the graph even for small $p$;
in particular, already for $p \geq 0.1$,  
the reach of every ad campaign $j$ matches the budget \BUDGET[j].
The curve for \NetHEPT looks similar, but for different reasons.
Here, the reason appears to be that \NetHEPT is sparse and exhibits
little competition from the start, which is confirmed by the fact that
\RP never drops below 15.
Here, the advertisers' budgets \BUDGET[j] are only fully extracted
around $p=0.25$, but the decentralized nature of the graph ensures
lack of harmful competition much earlier,
which is why \RP saturates around $p=0.1$.

For \Facebook and \DBLP, the two graphs exhibiting more competition,
the minimum is attained at $p=0.04$.
Notice that the minimum value of \RP for these data sets is
significantly smaller and the saturation of $\RP = 20$ happens later
than for the other networks.
This behavior again shows the stronger competition.
Specifically for \DBLP, even at $p=0.3$, the network only extracted
about half of the sum of all advertiser budgets.

\subsection{Scalability and Parallelization} \label{sec:scalability}
\label{sec:app-experiments}
We evaluated the scalability of the Greedy algorithm on the following
two large datasets:

\begin{itemize}
\item \Pokec \cite{konect} is a friendship network from the Slovakian 
  social network Pokec, containing 1,632,803 individuals (nodes), and
  30,622,564 directed edges (directed friendships).
\item \Flixster \cite{konect} is a social network of a movie rating 
  site where people with similar cinematic taste can connect with each
  other. It contains 2,523,386 individuals and 7,918,801 (undirected)
  edges.
\end{itemize}

We evaluated scalability by measuring the runtime as a function of
the number of advertisers $m$.
The LP Rounding algorithm does not scale to such large datasets;
its limit is a few tens of thousands of nodes and edges.



In order to utilize all 36 cores on the server,
we parallelized the RR set construction in the Greedy algorithm, which 
contributes most to the overall execution time.
Figure~\ref{fig:scalability} shows that the parallel algorithm can 
compute seed sets on the large graphs in scenarios with large numbers
of competing advertisers,
even when the sequential algorithm fails to terminate in a reasonable
amount of time.
The parallel Greedy algorithm accelerates RR-set construction by a
factor about 22 on average,
and is overall about 12 times faster than the sequential version.

Due to the large graph sizes and extremely skewed degree 
distributions in \Pokec and \Flixster,
the highly influential nodes had significantly more impact,
and the overall payoff was up to two orders of magnitude larger than
for the \Facebook, \Advogato, \NetHEPT and \DBLP networks.
This resulted in slower seed set selection,
which became a significant factor when the RR-set construction was
parallelized.
As a result, we see that the runtime of our parallel implementation
(especially in \Pokec) scales sub-linearly with respect to $m$ when there
are more than a dozen advertisers.

\begin{figure}
\includegraphics[width=\linewidth]{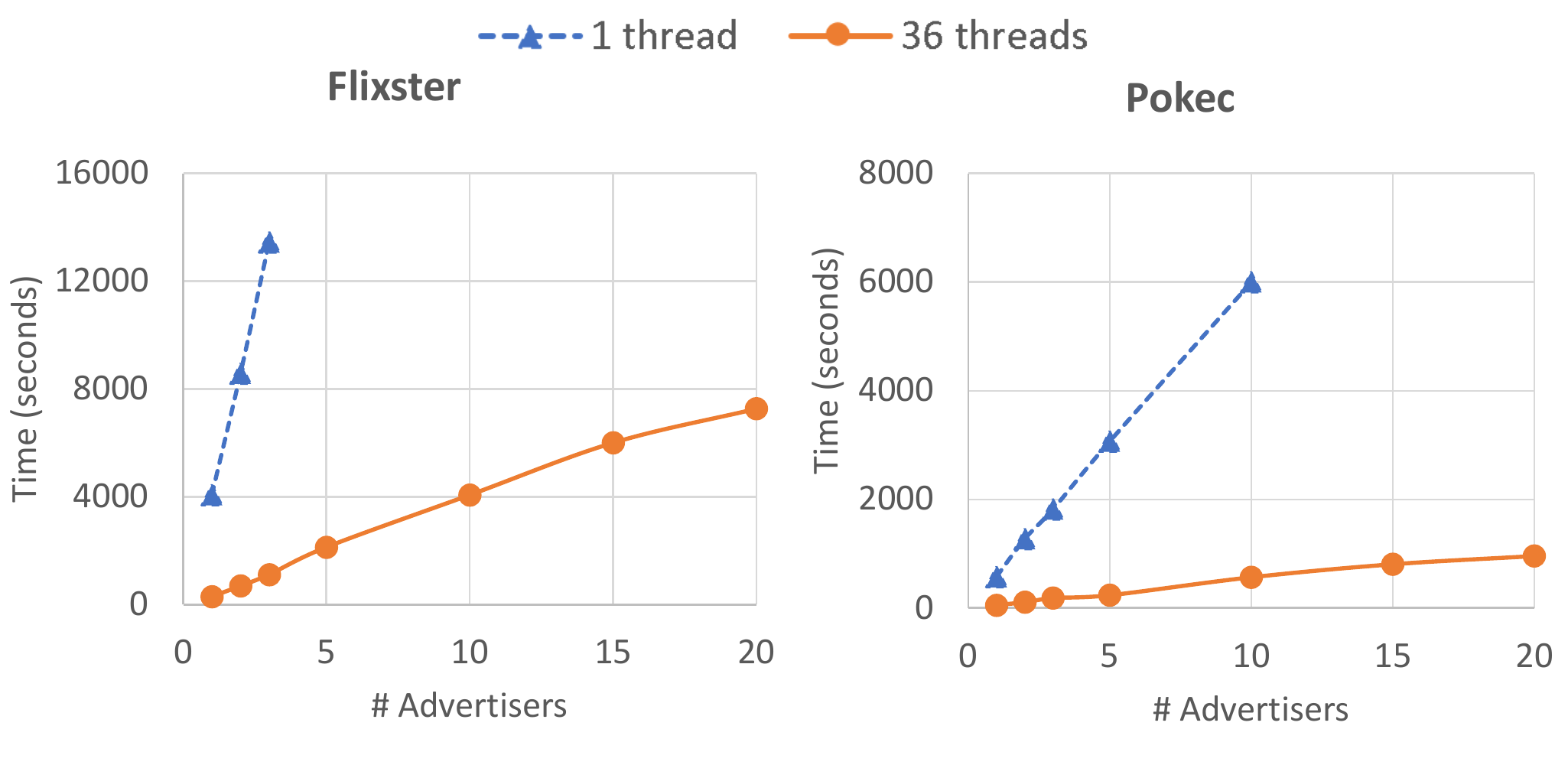}
\caption{Running times of the sequential and parallel Greedy algorithm
  on the \Flixster and \Pokec datasets.} \label{fig:scalability}
\end{figure}

\section{Conclusions and Future Work} \label{sec:conclusions}
We explored the problem of maximizing the host's payoff for a budgeted
multi-advertiser setting with constraints on the ad exposure of
each node and each advertiser's seed set size.  
We showed that under a very general class of influence models,
this problem can be cast as maximizing a submodular function subject
to two matroid constraints.
It can therefore be approximated to within essentially a factor \half
using an algorithm due to
Lee et al.~\cite{lee:sviridenko:vondrak:multiple-matroids};
this improves on the \third-approximation guarantee of
\cite{datta:majumder:shrivastava}.
When there are no constraints on individual advertisers' seed set
sizes, only on the total number of seeds,
the constraints form just a single truncated partition matroid.
Therefore, the algorithm of Vondr\'{a}k et
al.~\cite{calinescu:chekuri:pal:vondrak,vondrak:submodular-welfare}
gives $(1-\frac{1}{\e})$-approximation algorithm, 
while a simple greedy
algorithm~\cite{fischer:nemhauser:wolsey}
achieves a \half-approximation,
as observed by Tang and Yuan \cite{tang2016optimizing}. 
For the special case of the Independent Cascade model,
the Reverse Reachability Technique~\cite{borgs:brautbar:chayes:lucier,%
tang:shi:xiao:martingale,tang:xiao:shi:near-optimal}
can be combined with a linear program and rounding due to
\cite{calinescu:chekuri:pal:vondrak:ipco,gandhi:khuller:parthasarathy:srinivasan:dependent}
to yield a more efficient $(1-1/\e)$-approximation algorithm,
even under both types of constraints.

Our experiments show that in practice, the greedy algorithm slightly
outperforms the LP-rounding based algorithm, despite its worse approximation
guarantee.
They also reveal that competition between advertisers leads to a loss
in payoff for the host when the advertisers' influence networks are
more similar, and particularly so when the networks have few highly
influential nodes --- in that case, these nodes cannot simultaneously
be used as seeds for all advertisers, and 
the average exposure per advertiser becomes significantly lower than
the individual exposures of the advertisers on isolated networks.

Our work was in part motivated by considering a more ``natural''
objective function than the notion of regret from \cite{aslay2015viral}.
As discussed in Section~\ref{sec:introduction}, a less Draconian way
to penalize excess exposure is to linearly penalize over-exposure.
In fact, Tang and Yuan \cite{tang2016optimizing} analyze exactly such
an objective, simply subtracting the excess exposure from the revenue.
They provide a $\frac{1}{4}$-approximation algorithm for this objective.
Their objective function can be further generalized by scaling this
penalty term by a user-controlled constant.
This allows a host to trade off between the two very different parts
of the objective (revenue vs.~free exposure),
since they are not directly comparable.
In additional work not included here,
we show that even under this more general model,
modifications of the greedy algorithm and LP rounding algorithm
still yield constant-factor approximation guarantees.

There are several natural directions for future work.
Perhaps most directly, the LP-rounding based algorithm only satisfies
the budget constraint in expectation;
indeed, we have shown a large integrality gap for the LP.
It is a natural question whether a $(1-1/\e)$-approximation guarantee
can be obtained while always satisfying the budget constraints,
and without invoking the heavy Continuous Greedy machinery discussed
in Section~\ref{sec:general}. 

We believe that the efficiency, scalability and solution quality of our algorithms,
especially the LP-Rounding algorithm (Section~\ref{sec:LPRound}), can
be further improved, e.g., by adapting techniques such as the ensemble
approach for single-product influence maximization recently proposed
by G\"{u}ney~\cite{guney2019efficient}.
This method uses multiple batches of a small number of RR sets to
generate several candidate seed sets.
The key idea here is that the probability of finding an optimal seed
increases exponentially in the number of batches, even though smaller  
RR sets increase the variability of individual candidate sets.

An interesting empirical study would be to evaluate to what extent
influential nodes in one advertiser's network are also influential in other
advertisers' networks.
Within typical social networks, one would expect significant
differences based on individuals' expertise;
on the other hand, the use of celebrities to endorse products entirely
outside their realm of expertise shows that humans appear willing to
project expertise in one area on other areas as well.

\subsubsection*{Acknowledgments}
We would like to thank Ajitesh Srivastava for helpful technical discussions.
David Kempe was supported in part by
NSF Grant 1619458 and ARO MURI grant 72924-NS-MUR.

\bibliographystyle{plain}
\bibliography{davids-bibliography/names,davids-bibliography/conferences,davids-bibliography/bibliography,davids-bibliography/publications,sample}

\newpage

\appendix

\section{Review of Basic Concepts} \label{sec:review}
In this section, we provide definitions and descriptions of several
concepts we believe to be standard,
but which some readers may want to review.

\subsection{Matroids and (truncated) Partition Matroids}
\label{sec:matroid-definitions}

\begin{definition}[Matroid]
  \label{def:matroid}
  A \emph{matroid} is a non-empty, downward-closed%
\footnote{That is, $\mathcal{I} \neq \emptyset$, and
if $S \in \mathcal{I}$, then $S' \in \mathcal{I}$ whenever
$S' \subseteq S$.}
set system $(X,\mathcal{I})$,
with the following \emph{exchange property}:
if $S, S' \in \mathcal{I}$, and $|S| < |S'|$, then there exists some
element $x \in S' \setminus S$ such that
$S \cup \SET{x} \in \mathcal{I}$.
The sets in $\mathcal{I}$ are called \emph{independent sets}.
\end{definition}

\begin{definition}[Partition Matroid, Truncation]
  \label{def:partition-matroid}
A partition matroid is defined as follows:
Let $X_1, X_2, \ldots, X_p$ be a disjoint partition of $X$,
and $b_1, b_2, \ldots, b_p$ be non-negative integers.
A set $S$ is independent (i.e., in $\mathcal{I}$) iff
$|S \cap X_i| \leq b_i$ for all $i$.

A \emph{truncation} of a matroid $(X,\mathcal{I})$ is obtained by
replacing $\mathcal{I}$ with $\Set{S \in \mathcal{I}}{\SetCard{S} \leq k}$,
i.e., by discarding all sets of size exceeding $k$.
It is easy to see that any truncation of a matroid is again a matroid.
\end{definition}

\subsection{The Independent Cascade Model} \label{sec:IC-definition}

The Independent Cascade (IC) Model \cite{goldenberg:libai:muller:complex,%
  goldenberg:libai:muller:talk,InfluenceSpread}
is defined as follows.
(We give the generalization to multiple ads directly here.)
For each ad $j$ and (directed) edge $e=(u,v)$,
there is a known probability \EdgeProb[j]{e} with which $u$ will
succeed in activating $v$. 
Starting from a seed set \Active[j]{0},
in each round $t$, each newly activated (for ad $j$) node $u$ can make
one attempt to activate each currently inactive (for $j$) neighbor $v$.
The attempt succeeds with probability \EdgeProb[j]{u,v},
independently of other activation attempts.
If at least one activation attempt (for $j$) on $v$ is successful,
$v$ will become active for $j$ at time $t+1$, i.e., become part
of \Active[j]{t+1}.
As proved, e.g., in \cite{InfluenceSpread}, the Independent Cascade
Model is a special case of the General Threshold Model from
Section~\ref{sec:problem-statement}, by setting
$\InfFun[j]{v}{S} = 1-\prod_{u \in S} (1-\EdgeProb[j]{u,v})$.

A very useful alternative view of the Independent Cascade Model was
first shown in \cite{InfluenceSpread}, and heavily used in subsequent
work: generate graphs \GRAPH[j] by including each edge $e$ in
\GRAPH[j] independently with probability \EdgeProb[j]{e}.
Then, the distribution of nodes activated by ad $j$ in the end,
when starting from the set \Active[j]{0},
is the same as the distribution of nodes reachable from \Active[j]{0} in
the random graph \GRAPH[j].


\section{Number of Reverse Reachable Sets vs.~Estimation Error} \label{sec:RR-set-experiments}
An important part of all algorithms for the optimization problem is
being able to evaluate the objective function.
Doing so exactly is \#P-complete
\cite{chen:yuan:zhang:scalable,chen:wang:wang:prevalent}.
Good approximations can be obtained by using Reverse Reachable (RR)
sets, as described in Section~\ref{sec:IC}.
However, the theoretical bounds from Section~\ref{sec:IC} that
\emph{guarantee} a good approximation of the objective, while
polynomial, are sufficiently large that algorithms using this many
RR sets would not scale well.
The goal of this section is to experimentally evaluate how many RR
sets yield good estimates of the objective function
\emph{in practice}.

To do so, we ran experiments on the four networks from
Section~\ref{sec:experiments}: \Facebook, \Advogato,
\NetHEPT and \DBLP.
For every advertiser, we created independent influence networks as
described in Section~\ref{sec:experiments}.
In our experiments, we varied the number of advertisers $m$ from $1$ to
$20$ and the number of RR sets from $mn$ to $100mn$. 
For simplicity, we assumed that each advertiser $j$'s value function is
simply the (unscaled, and unbudgeted) number of nodes exposed to $j$'s
ad.

The main idea is to obtain several independent estimates of the
objective function by drawing independent samples of RR sets.
When the number of RR sets in each sample is large enough,
the estimates should be similar to each other according to different
metrics, such as absolute error, variance, etc.

\subsection{Absolute Error} \label{sec:err}
In the first set of experiments, we computed reference seed sets
$\Active[j]{0}$ by running the Greedy algorithm
using $200mn$ RR sets for estimates.%
\footnote{Whether these were actually good seed sets is secondary,
as we used them to measure the quality of the \emph{estimation} only.
Running the Greedy algorithm ensured that we avoided trivial cases,
such as accidentally choosing seed sets with extremely small
influence.}

Having computed the seed sets \Active[j]{0}, we estimated their payoff
(for all advertisers) 
using different numbers of RR sets.
The absolute error of an estimate is the sum, over all advertisers, of
the absolute difference of the estimate using the $200mn$ RR sets
and the estimate of the newly drawn RR sets.


In Figure~\ref{fig:l1err}, we plot the absolute error for all four
networks, as we vary the number of RR sets drawn.
Each plot was obtained by averaging the absolute error over 100
independent draws of the evaluation RR sets.
We observe that for almost all network sizes $n$ and numbers of
advertisers $m$, the error is less than 2\% 
once $\rho^{(j)}\geq 10n$ for all $j$;
at that point, the error curve also nearly flattens out.

\begin{figure*}[htbp]
\includegraphics[width=\textwidth]{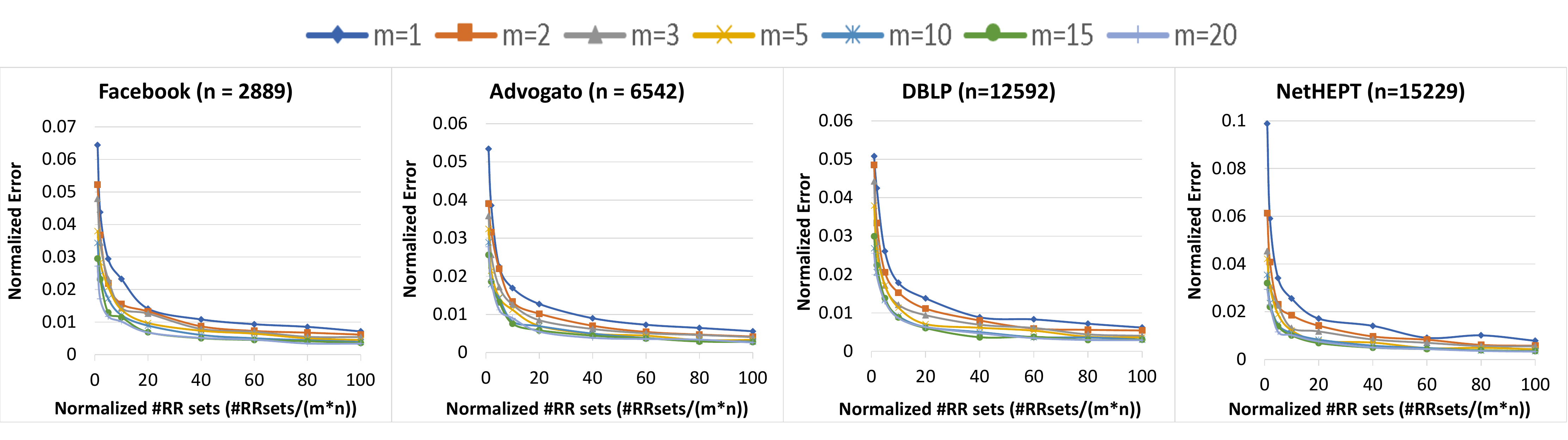}
\caption{Absolute error in payoff estimation for different numbers of
  RR sets and advertisers $m$. The $x$-axis plots the ratio between the
  total number of RR sets and $mn$.}
  \label{fig:l1err}
\end{figure*}

\subsection{Standard Deviation} \label{sec:stddev}
As a second error measure, we considered the standard deviation.
The experimental setup was identical to the one in Section \ref{sec:err}.
The only difference is that instead of the total absolute error,
in Figure~\ref{fig:stddev}, we plot the standard deviation of the
estimations, normalized by the estimated mean.

\begin{figure*}[htbp]
\includegraphics[width=\textwidth]{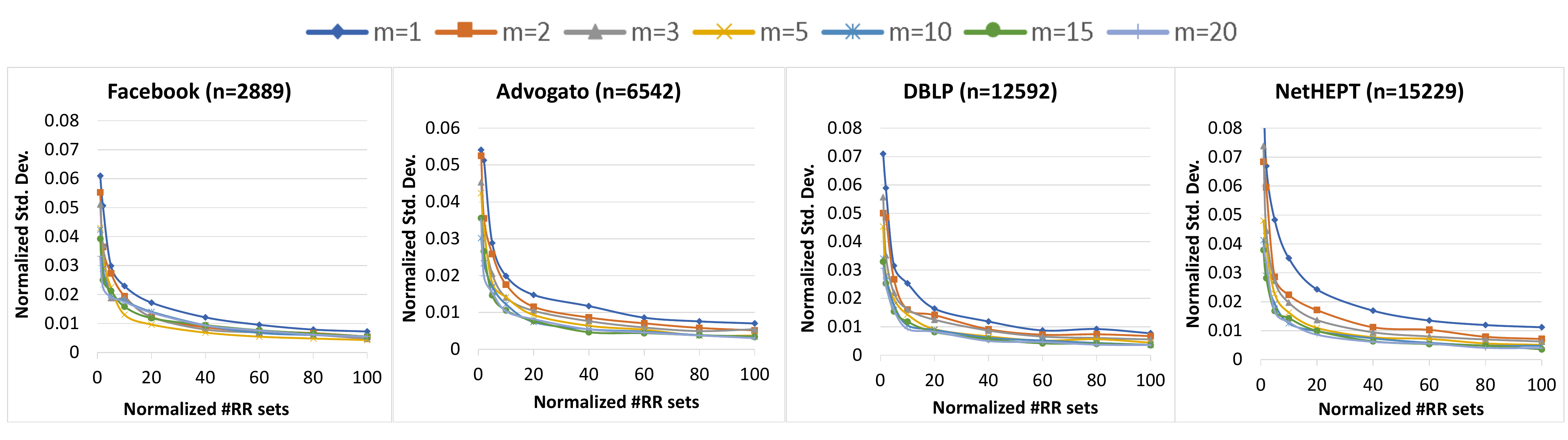}
\caption{Standard deviation of the estimated payoff, normalized by the
  mean, for different numbers of RR sets and advertisers $m$. The
  $x$-axis plots the ratio between the total number of RR sets and $mn$.}
  \label{fig:stddev}
\end{figure*}

The plots are nearly identical to those in Figure~\ref{fig:l1err},
confirming that our conclusions about estimation errors are robust to
the specific error measure used.

\begin{figure*}[htbp]
\includegraphics[width=\textwidth]{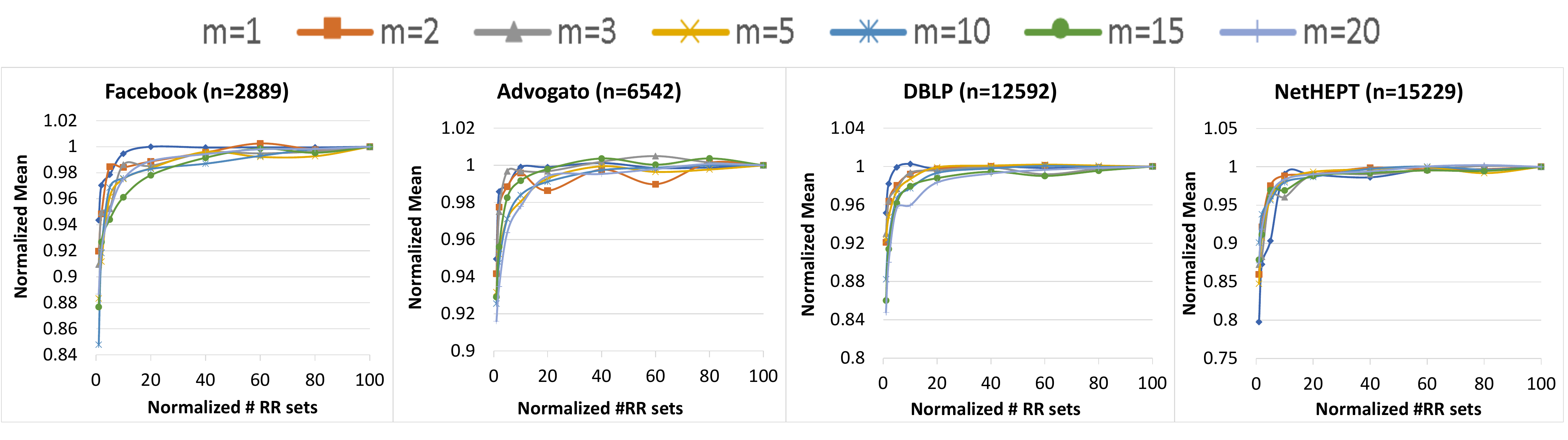}
\caption{Average payoff of selected seed set normalized by the payoff of seed set obtained with very large number of RR sets $(100mn)$. The x-axis plots the ratio of total number of RR sets and $m*n$.}
  \label{fig:mean}
\end{figure*}

\subsection{Impact on Optimization} \label{sec:mean}
In the third set of experiments, we analyzed the impact of the number of RR
sets on the quality of the seed sets selected by the Greedy algorithm.
Unlike in Sections~\ref{sec:stddev} and \ref{sec:err}, the Greedy
algorithm was now run with a limited number of RR samples for
estimation.
(The number is specified as the number of RR sets per
advertiser.)
In other words, the Greedy algorithm used estimates of higher
variance, and thus was expected to produce suboptimal seed sets.
To evaluate the performance, we normalized the objective value against
the objective value that a Greedy algorithm using $100mn$ RR sets
would obtain.

Figure~\ref{fig:mean} plots the quality of the Greedy algorithm,
normalized by the average payoff.
Plots are shown for all four networks and different values of $m$,
and are averaged over 100 runs.
We observe that with a small number of RR sets, such as $\rho^{(j)}=n$,
the quality of the selected seed sets can be quite poor,
and the average payoff obtained can be 10--20\% less than what can be
achieved using a larger number of RR sets, such as $\rho^{(j)}=100n$.
This is because with small $\rho^{(j)}$,
the estimation errors are high;
thus, the Greedy algorithm may overestimate the objective value of
some seed sets and underestimate the objective value of other seed sets.%
\footnote{Recall that payoffs were estimated with Reverse Reachability
  sets by computing the fraction of RR sets that contain at least one
  of the nodes from the seed set.}


For almost all cases, when $\rho^{(j)}=10n$, the payoff achieved was
at least 98\% of the payoff with $100n$ RR sets per advertiser.
As in Sections~\ref{sec:err} and \ref{sec:stddev},
the curves nearly flatten out for $\rho^{(j)} \geq 10n$.
Hence, we infer that empirically, about $10n$ RR sets per advertiser
are sufficient to obtain very accurate estimates of the objective
function, and to guarantee that the greedy algorithm performs almost
as well as with much larger numbers of RR sets.


\end{document}